\documentclass{article}[11pt]  
\usepackage{amsmath}
\usepackage{amssymb}
\usepackage{amsthm}
\usepackage{ifpdf}
\usepackage{wrapfig}
\usepackage{fullpage}
\usepackage{cite}
\usepackage{subcaption}
\usepackage{float}
\usepackage{enumitem}
\usepackage{graphicx}
\usepackage[noend,ruled]{algorithm2e}
\usepackage{pdftexcmds}

\ifpdf

  \usepackage[pdftex]{epsfig}
  \usepackage[pdftex]{hyperref}

\else

    \usepackage[dvips]{epsfig}
    \newcommand{\href}[2]{#2}

\fi

\theoremstyle{definition}
\newtheorem{theorem}{Theorem}[section]
\newtheorem{lemma}[theorem]{Lemma}

\newtheorem{definition}[theorem]{Definition}
\newtheorem{problem}[theorem]{Problem}

\usepackage[table]{xcolor}
\definecolor{greyh}{HTML}{DDDDDD}

\def\compactify{\itemsep=0pt \topsep=0pt \partopsep=0pt \parsep=0pt}
\let\latexusecounter=\usecounter
\newenvironment{itemize*}
  {\def\usecounter{\compactify\latexusecounter}
   \begin{itemize}}
  {\end{itemize}\let\usecounter=\latexusecounter}
\newenvironment{enumerate*}
  {\def\usecounter{\compactify\latexusecounter}
   \begin{enumerate}}
  {\end{enumerate}\let\usecounter=\latexusecounter}
\newenvironment{description*}
  {\begin{description}\compactify}
  {\end{description}}

\newif\ifflex
\flexfalse

\begin{document}

\title{Unique Assembly Verification in Two-Handed Self-Assembly\thanks{This research was supported in part by National Science Foundation Grant CCF-1817602.}
}

\author{David Caballero \and Timothy Gomez \and Robert Schweller \and Tim Wylie}

\date{}
\clearpage\maketitle
\thispagestyle{empty}

\vspace*{-.5cm}
\begin{center}
Department of Computer Science\\University of Texas Rio Grande Valley \\Edinburg, TX 78539-2999, USA\\
\{david.caballero01, timothy.gomez01, robert.schweller, timothy.wylie\}@utrgv.edu
\end{center}

\begin{abstract}
One of the most fundamental and well-studied problems in Tile Self-Assembly is the Unique Assembly Verification (UAV) problem. This algorithmic problem asks whether a given tile system uniquely assembles a specific assembly. The complexity of this problem in the 2-Handed Assembly Model (2HAM) at a constant temperature is a long-standing open problem since the model was introduced. Previously, only membership in the class coNP was known and that the problem is in P if the temperature is one ($\tau=1$). The problem is known to be hard for many generalizations of the model, such as allowing one step into the third dimension or allowing the temperature of the system to be a variable, but the most fundamental version has remained open. 

In this paper, we prove the UAV problem in the 2HAM is hard even with a small constant temperature ($\tau = 2$), and finally answer the complexity of this problem (open since 2013).  Further, this result proves that UAV in the staged self-assembly model is coNP-complete with a single bin and stage (open since 2007), and that UAV in the q-tile model is also coNP-complete (open since 2004). We reduce from Monotone Planar 3-SAT with Neighboring Variable Pairs, a special case of 3SAT recently proven to be NP-hard. We accompany this reduction with a positive result showing that UAV is solvable in polynomial time with the promise that the given target assembly will have a tree-shaped bond graph, i.e., contains no cycles. We provide a $\mathcal{O}(n^5)$ algorithm for UAV on tree-bonded assemblies when the temperature is fixed to $2$, and a $\mathcal{O}(n^5\log \tau)$ time algorithm when the temperature is part of the input.
\end{abstract}
\newpage

\section{Introduction}

Since the inception of tile self-assembly~\cite{Winf98}, one of the most important algorithmic questions has been determining if a given tile system uniquely self-assembles into a specific assembly structure.  This basic algorithmic question, termed the Unique Assembly Verification (UAV) problem, is fundamental for efficiently checking if a designed tile system acts as intended, and is tantamount to the design of an efficient simulator for a tile self-assembly model. Thus, UAV has been a central question for every self-assembly model.

Although many different self-assembly models have been proposed in order to simulate different laboratory or experimental setups, two premiere models have emerged as the primary foci of study.  First, is the seeded Abstract Tile Assembly Model (aTAM)~\cite{Winf98}, in which singleton tiles attach one by one to a growing seed if sufficient bonding strength exists based on glue types of attaching tiles.  This model has had many foundational results in recent years showing the limits related to intrinsic universality and program-size complexity \cite{pumpability,t1notiu}. The second model is the hierarchical Two-Handed Tile Assembly Model (2HAM)~\cite{2HABTO}, where any two producible assemblies may be combined (one in each of two hands) to create a new producible assembly provided there is sufficient bonding strength between the two pieces. Many  foundational results that are known for the aTAM are still open for the 2HAM.

The 2HAM has been shown to be more powerful than the aTAM in its ability to build infinite fractal patterns~\cite{CFH2014SSA,HO2017SAS}, its program-size efficiency for finite shapes~\cite{CANNON202150}, and its running-time efficiency for the self-assembly of finite shapes~\cite{CheDotSICOMP2017}.  While the aTAM has a polynomial time solution to the UAV problem~\cite{ACGHKMR02}, allowing for the production of efficient simulators~\cite{Xgrow, pyTAS,versatile}, the complexity of UAV in the 2HAM has remained a long-standing open problem in the field. The 2HAM appeared formally in 2013 \cite{2HABTO}, but was essentially defined in staged self-assembly \cite{demaine2008staged} (2007), and a seeded version of the 2HAM appears as the \emph{multiple tile} model in \cite{ACG2005CGM} (2004).  UAV has been open for all of these models, and our coNP-complete result for UAV in the 2HAM proves that UAV with a single bin and single stage in the staged model is coNP-complete, and that UAV in the multiple tile model is also coNP-complete with polynomial-sized pieces, thus answering both of these long-standing open questions. See~\cite{Patitz2014,woods2015intrinsic,doty2012theory,drmaurdsa} for surveys and applications of self-assembly theory.

\paragraph{Previous work on UAV.}
A number of results have pushed closer to resolving the complexity of UAV in the 2HAM.  One of the first results showed that the simpler problem of determining if a given assembly was at least produced (i.e., built but possibly along with other different assemblies) is polynomial time solvable~\cite{prod2HAM}, which serves as a key step in showing that UAV resides within the class coNP.  Another result augmented the basic 2HAM model to 3 dimensions and showed coNP-hardness for the 3D 2HAM \cite{2HABTO}.  A recent result focused on 2D, but allowed the temperature threshold, a parameter that determines how much glue strength is required for assemblies to stick together, to be a variable input to the UAV problem (as opposed to a fixed constant value), and showed coNP-completeness in this scenario~\cite{Schweller:2017:DNA}.

Other approaches considered the allowance of initial assemblies consisting of small prebuilt assemblies, as opposed to only initial singleton tiles, and showed UAV becomes coNP$^\text{NP}$-complete with this extension \cite{Caballero:2021:ARXIV}.
Alternately, the inclusion of a negative force glue, even without detachments, has also been shown to imply coNP-completeness in the aTAM \cite{Cantu:2021:A}.
Another generalization of the 2HAM allows for up to $k$ hands to create new assemblies, instead of just two, causing the problem to become either coNP-complete or PSPACE-complete, depending on the encoding of the variable $k$~\cite{Caballero:2021:UCNC}.  An even more powerful generalization of the 2HAM is the \emph{staged} model~\cite{demaine2008staged}, in which multiple distinct stages of self-assembly are considered.  Within the staged model, UAV becomes coNP$^\text{NP}$-hard after 3 stages, and PSPACE-complete in general~\cite{schweller2019verification,Caballero:2021:ESA}.  Thus, for nearly every way in which the 2HAM has been extended, a corresponding hardness reduction has been found.  Yet, the original question of UAV in the 2HAM has remained open.

\begin{table}[t]
	\centering
	
	\renewcommand{\arraystretch}{1.5}
	\begin{tabular}[b]{| c | c | c | c | c | }
		\hline 
		 \textbf{Shape} & \textbf{Dimensions} & \textbf{Temperature} & \textbf{Complexity} & \textbf{Reference}\\  \hline
		General & $2$ & $1$ & $\mathcal{O}(|A||T|\log |T|)$ & \cite{prod2HAM} \\ 
		General & $3$ & $2$ & coNP-complete &  \cite{2HABTO} \\ 
		General & $2$ & $\tau$ & coNP-complete & \cite{Schweller:2017:DNA}  \\ \hline 
		\rowcolor{greyh} 
		General & 2 & 2 & coNP-complete &  Thm. \ref{thm:coNPC} \\ 
		\rowcolor{greyh} 
		Tree & 2 & $\tau$ & $\mathcal{O}(|A|^5 \log \tau)$ &  Thm. \ref{thm:tauTree} \\ \hline
	\end{tabular}
	\caption{Known Results for the Unique Assembly Verification Problem in the 2HAM and the results presented in this paper. $|A|$ is the size of the target assembly, $\tau$ is the temperature of the system, and $|T|$ is the number of tile types in the system. Under the \textbf{Temperature} column, $\tau$ indicates that the temperature may be included as part of the input. }
\label{tab:results}
\end{table} 

%

\paragraph{Our Contributions.}
We show that UAV in the 2HAM is coNP-complete within the original model (2-dimensional, constant bounded temperature parameter, singleton tile initial assemblies), thus resolving the long-standing open problem of UAV in the 2HAM.  Further, this proves that UAV in the staged model with a single bin and stage is coNP-complete, and that UAV in the q-tile/multiple-tile model with polynomial-sized pieces is coNP-complete. We augment this result with a positive result for the special case of tree-shaped assemblies, providing a $\mathcal{O}(|A|^5 \log \tau)$ time solution for UAV in this case (where $|A|$ is the size of the assembly) even if $\tau$ is included as part of the input.

Our results are highlighted in Table \ref{tab:results} along with other known results for UAV in the 2HAM. To show coNP-hardness for UAV we construct an explicit polynomial-time reduction from Monotone Planar 3-SAT with Neighboring Variable Pairs (MP-3SAT-NVP).  This reduction takes inspiration from the recent break-through proof that MP-3SAT-NVP is NP-hard and its use to prove that the connected-assembly-partitioning problem with unit squares is NP-hard~\cite{Agarwal:2021:SODA}. 
For our tree UAV algorithm, we utilize a cycle decomposition approach over possible produced assemblies combined with dynamic programming.

\paragraph{Overview.} The paper is structured as follows. Section \ref{sec:definitions} formally defines the model, the UAV problem, important definitions, and some small examples. Section \ref{sec:uavhard} has the reduction proving UAV in the 2HAM is coNP-hard. Due to the numerous intricate details related to the proof, the section is broken up into several subsections explaining different aspects of the reduction. Section \ref{sec:tree} then gives the algorithms for solving UAV for tree-bonded assemblies. Section \ref{sec:conclusion} then concludes the paper with a summary and future work.

\section{Definitions}\label{sec:definitions}

In this section we overview the basic definitions related to the two-handed self-assembly model and the verification problems under consideration.

\paragraph{Tiles.}
A \emph{tile} is a non-rotating unit square with each edge labeled with a \emph{glue} from a set $\Sigma$.
Each pair of glues $g_1, g_2 \in \Sigma$ has a non-negative integer \emph{strength} ${\rm str}(g_1, g_2)$.

\paragraph{Configurations.}
A \emph{configuration} is a partial function $\tilde{A} : \mathbb{Z}^2 \rightarrow T$ for some set of tiles $T$, i.e. an arrangement of tiles on a square grid.
For a configuration $\tilde{A}$ and vector $\vec{u} = \langle u_x, u_y \rangle$ with $u_x, u_y \in \mathbb{Z}^2$, $\tilde{A} + \vec{u}$ denotes the configuration $\tilde{A} \circ f$, where $f(x, y) = (x + u_x, y + u_y)$.
For two configurations $\tilde{A}$ and $\tilde{B}$, $\tilde{B}$ is a \emph{translation} of $\tilde{A}$, written $\tilde{B} \simeq \tilde{A}$, provided that $\tilde{B} = \tilde{A} + \vec{u}$ for some vector $\vec{u}$.

\paragraph{Bond graphs, and stability.}
For a given configuration~$\tilde{A}$, define the \emph{bond graph}~$G_{\tilde{A}}$ to be the weighted grid graph in which each element of~${\rm dom}(\tilde{A})$ is a vertex, and the weight of the edge between a pair of tiles is equal to the strength of the coincident glue pair.
A configuration is said to be \emph{$\tau$-stable} for a positive integer~$\tau$  if $G_{\tilde{A}}$ is connected and  if every edge cut of $G_{\tilde{A}}$ has a weight of at least $\tau$. This means that the sum of the glue strengths along each cut is greater or equal to $\tau$.  An example bond graph for a small assembly is shown in Figure \ref{subfig:bond}.

\begin{figure}[t]
  \vspace*{-.2cm}
  \centering
 \begin{subfigure}[b]{.3\textwidth}
 	\centering
      \includegraphics[width=0.7\textwidth]{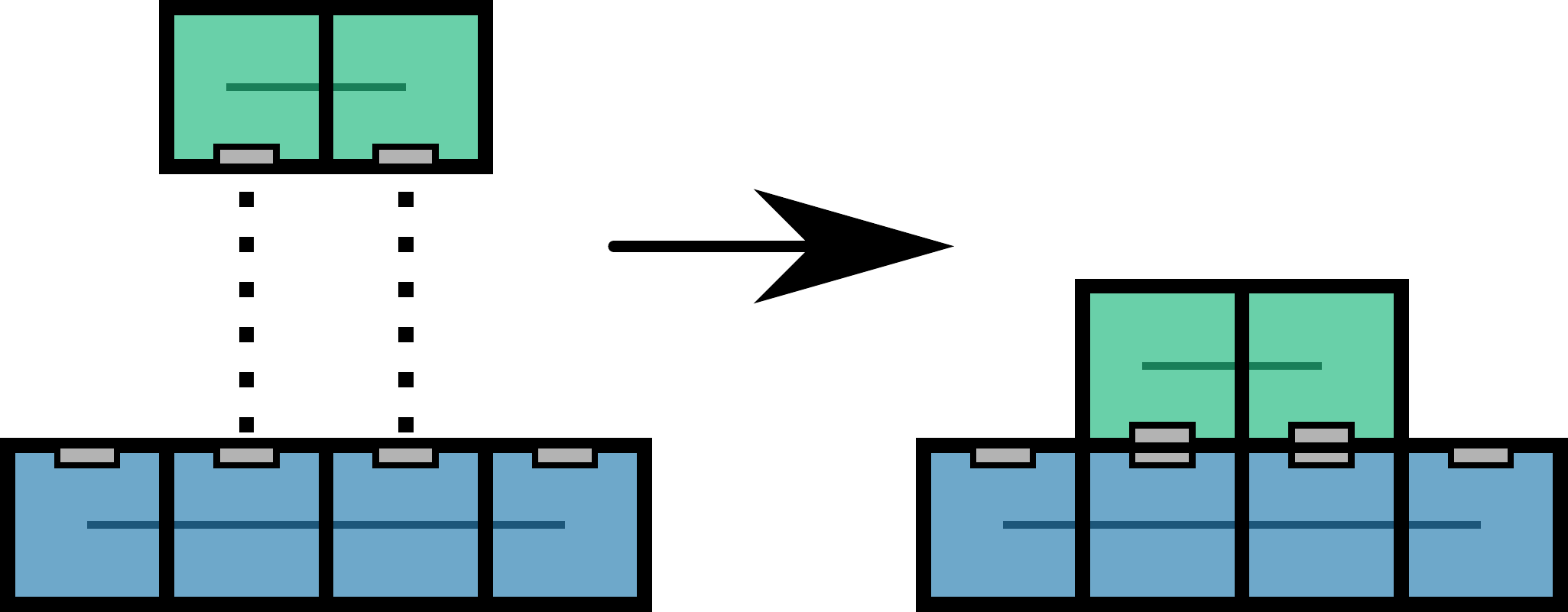}
   \caption{Cooperative Binding}\label{subfig:coop}
 \end{subfigure}
  \begin{subfigure}[b]{.25\textwidth}
 	\centering
 	\includegraphics[width=.4\textwidth]{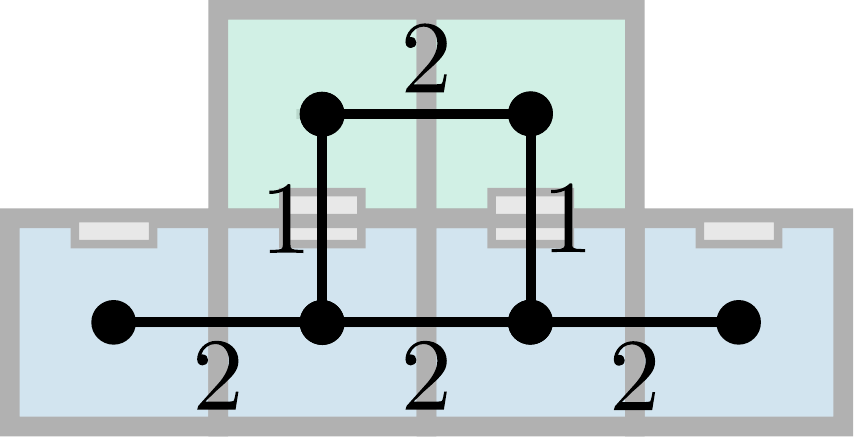}
 	\caption{Bond Graph} \label{subfig:bond}
 \end{subfigure}
 \begin{subfigure}[b]{.4\textwidth}
 	\centering
 	\includegraphics[width=.7\textwidth]{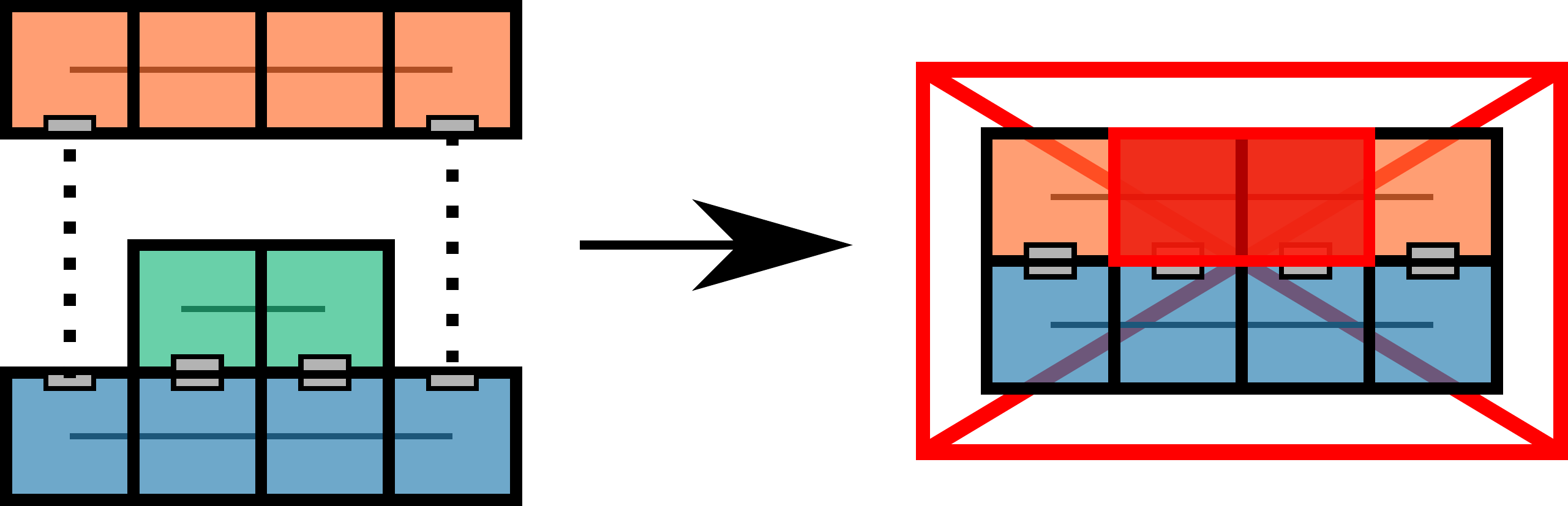}
 	\caption{Geometric Blocking}\label{subfig:geob}
 \end{subfigure}
  \vspace*{-.2cm}
  \caption{(a) Example of an attachment that takes places using cooperative binding at $\tau = 2$. We denote a glue strength of 1 with a rectangle and a glue of strength 2 with a solid line through the two tiles. Dotted lines between glues indicate that these tiles may attach to each other with the respective strength. Assume assemblies shown are $\tau$-stable unless stated otherwise. (b) The bond graph of the assembly showing that it is $\tau$-stable. (c) These two assemblies are not $\tau$-combinable since this would place two tiles at the same location. We say this is due to geometric blocking. 
   }\label{fig:att}
  \vspace*{-.2cm}
\end{figure}

\paragraph{Assemblies.}
For a configuration $\tilde{A}$, the \emph{assembly} of $\tilde{A}$ is the set $A = \{ \tilde{B} : \tilde{B} \simeq \tilde{A} \}$.
Informally an assembly $A$ is a set containing all translations of a configuration $\tilde{A}$. 
An assembly $A$ is a \emph{subassembly} of an assembly $B$, denoted $A \sqsubseteq B$, provided that there exists an $\tilde{A}\in A$ and $\tilde{B}\in B$ such that $\tilde{A} \subseteq \tilde{B}$. We define $|A|$ to be the number of tiles in a configuration of $A$. 

An assembly is \emph{$\tau$-stable} if the configurations it contains are $\tau$-stable.
Assemblies $A$ and $B$ are \emph{$\tau$-combinable} into an assembly $C$ if there exist $\tilde{A} \in A$, $\tilde{B} \in B$, and $\tilde{C} \in C$ such that
\begin{enumerate} \setlength\itemsep{0em}
    \item $\tilde{A} \cup \tilde{B} = \tilde{C}$. 
    \item $\tilde{A} \cap \tilde{B} = \varnothing$. 
    \item $\tilde{C}$ is $\tau$-stable. 
\end{enumerate}

Informally, two assemblies are $\tau$-combinable if there exists two configurations of the assemblies that may be combined resulting in a $\tau$-stable assembly without placing two tiles in the same location. 

Two assemblies combining or binding together is called an attachment. An attachment takes place using \emph{cooperative binding} if the two assemblies do not share a $\tau$-strength glue and instead use multiple weaker glues summing to $\tau$. An example of an attachment that takes place using cooperative binding can be seen in Figure \ref{subfig:coop}. 
If an attachment cannot take place because the two tiles would be placed in the same position, it is \emph{geometrically blocked}. Two assemblies whose attachment is geometrically blocked is shown in Figure \ref{subfig:geob}.

\paragraph{Two-handed Assembly.}
A two-handed assembly system (2HAM) is an ordered tuple $\Gamma = (T,\tau)$ where $T$ is a set of single tile assemblies and $\tau$ is a positive integer parameter called the \emph{temperature}.
For a system $\Gamma$, the set of \emph{producible} assemblies $P'_{\Gamma}$ is defined recursively as follows:
\begin{enumerate} \setlength\itemsep{0em}
    \item $T \subseteq P'_{\Gamma}$.
    \item If $A,B \in P'_{\Gamma}$ are $\tau$-combinable into $C$, then $C \in P'_{\Gamma}$.
\end{enumerate}
\vspace*{-.2cm}

A producible assembly is \emph{terminal} provided it is not $\tau$-combinable with any other producible assembly. Denote the set of all terminal assemblies of a system $\Gamma$ as $P_{\Gamma}$.
Intuitively, $P'_{\Gamma}$ represents the set of all possible assemblies that can self-assemble from the initial set $T$, whereas $P_{\Gamma}$ represents only the set of assemblies that cannot grow any further. Figure \ref{fig:example} shows a small 2HAM example with a single terminal assembly.

An \emph{Assembly Tree} for a 2HAM system $\Gamma =(T,\tau)$ is any rooted binary tree whose nodes are elements of $P'_\Gamma$, the leaves are single-tile assemblies from the set $T$, and the two children of any non-leaf node are $\tau$-combinable into their parent.  An assembly tree with root $A$ is said to be an assembly tree for assembly $A$. A small example is shown in Figure \ref{fig:astreet}.

\begin{figure}[t]
  \centering
 \begin{subfigure}[b]{.3\textwidth}
 	\centering
      \includegraphics[height=4cm]{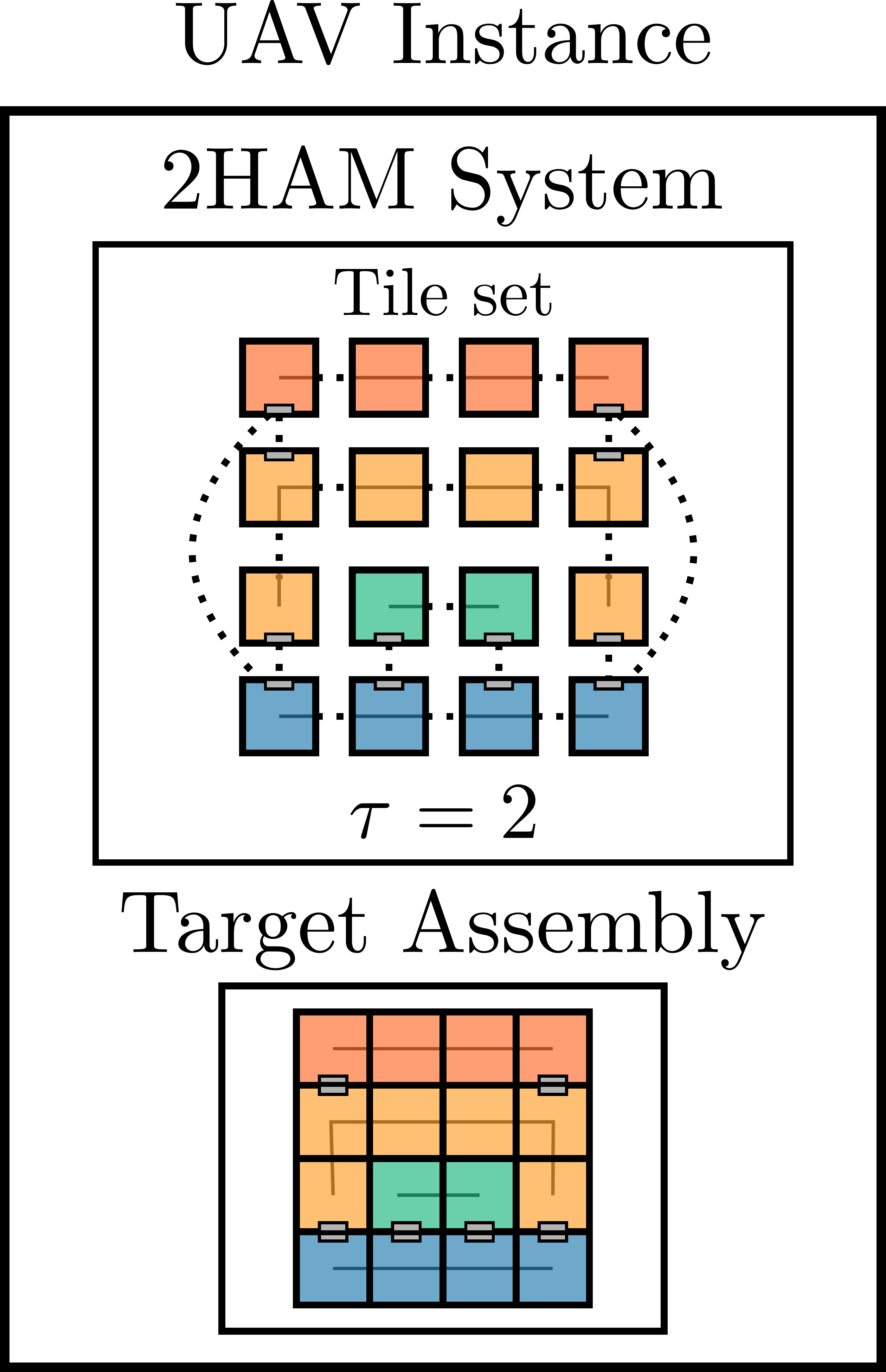}
   \caption{UAV Instance}
 \end{subfigure}
 \begin{subfigure}[b]{.42\textwidth}
 	\centering
 	\includegraphics[height=4cm]{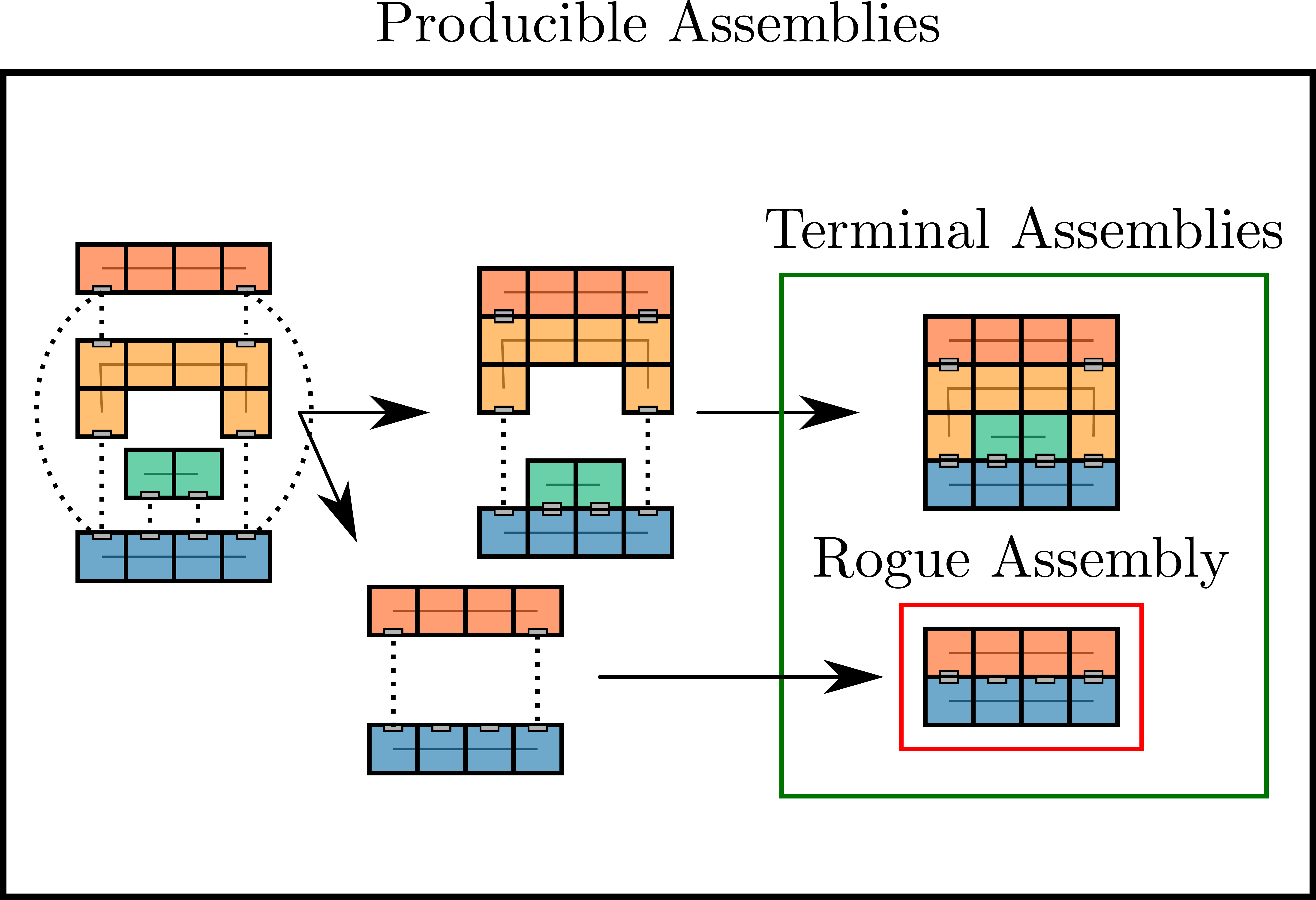}
 	\caption{Producible Assemblies}
 \end{subfigure}
 \begin{subfigure}[b]{.25\textwidth}
 	\centering
 	\includegraphics[width=.8\textwidth]{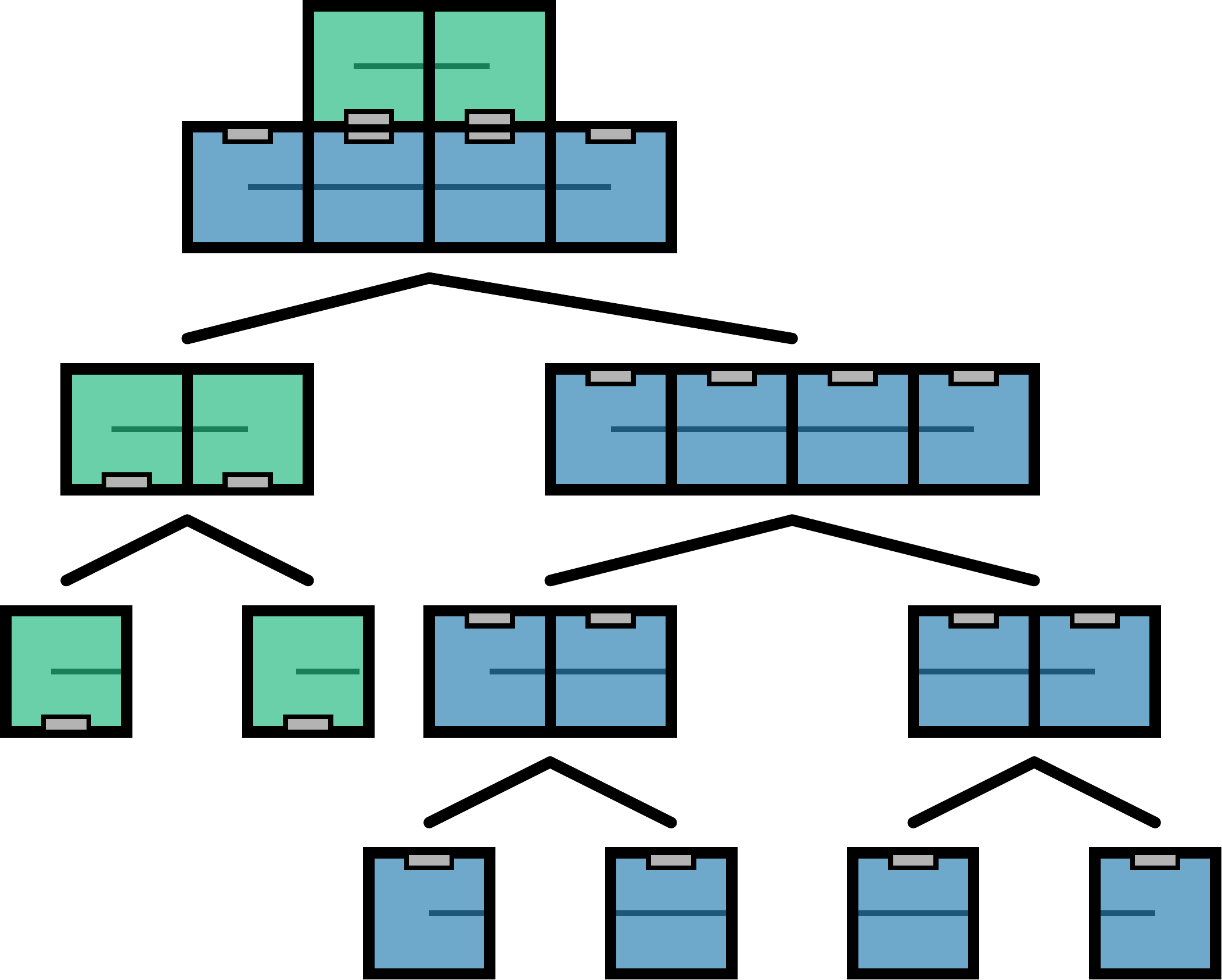}
 	\caption{Small Assembly Tree}\label{fig:astreet}
 \end{subfigure}
  \caption{(a) An example instance of the Unique Assembly verification problem. The input is the 2HAM system (tile set and temperature) and the target assembly. (b) The main producible assemblies of the 2HAM system (for clarity, not all subassemblies are shown). The target assembly is producible and terminal. However, there is also a produced assembly that is a rogue assembly (highlighted) since it not a subassembly of our target and it is terminal. 
  (c) A small example of an assembly tree for one of the producibles.
   }\label{fig:example}
\end{figure}

\paragraph{Unique Assembly.}
Intuitively, the unique assembly of $A$ means that any produced assembly can continue to grow until it becomes $A$, thus making $A$ the \emph{uniquely produced} assembly if the process is provided sufficient time to assemble. This means $A$ is the unique terminal assembly and all produced assemblies are subassemblies of $A$.  Formally we say a system $\Gamma$ uniquely produces an assembly $A$ if the following are true, 
\begin{enumerate} \setlength\itemsep{0em}
    \item $P_{\Gamma} = \{ A \}$
    \item For all $B \in P'_\Gamma$, $B \sqsubseteq A$
\end{enumerate}
\vspace*{-.2cm}

\begin{problem}[Unique Assembly Verification problem]
\textbf{Input:} A 2HAM system $\Gamma$, an assembly $A$.\\ \textbf{Output:} Does $\Gamma$ uniquely produce the assembly $A$?
\end{problem}

The Unique Assembly Verification problem (UAV) is the computational problem that asks to verify if an assembly is uniquely produced. A key concept used throughout this paper is a \emph{rogue assembly}. A rogue assembly is any producible assembly that breaks one of the conditions of unique assembly and serves as a proof that the instance of the UAV problem is false. 

\begin{definition}[Rogue Assembly]
Given an instance of UAV $(\Gamma, A)$, an assembly $R \sqsubseteq P'_\Gamma$ is a rogue assembly if $R \neq A$ and $R$ is not a subassembly of $A$. 
\end{definition}

We prove the following Lemma, which is used in the hardness reduction and the positive result. This lemma states that if the instance of UAV is false and all the tiles in $\Gamma$ are used to build $A$, then we may find a rogue assembly by only checking combinable subassemblies of $A$. 

\begin{lemma}
For an instance of UAV $(\Gamma, A)$ that is false, either there exist two assemblies $B, C$ such that $B, C \sqsubseteq A$ and $B$ and $C$ are $\tau$-combinable into a rogue assembly $R$, or there exists a single tile assembly $t \in P'_\Gamma$ that is not a subassembly of $A$. 
\label{lem:subRogue}
\end{lemma}
\begin{proof}
First since the instance of UAV is false we know there must exist some rogue assembly $R$. If $R$ is a single tile assembly the Lemma is true. If $R$ is not a single tile assembly we can walk through it's assembly tree to find the assemblies $B$ and $C$ which are both subassemblies of our target $A$. 

Consider an assembly tree of $R$ $\Upsilon_R$. Start by viewing the root, if it's two children are both subassemblies of $A$ then the rogue assembly $R$ satisfies the Lemma. If either of the children is also a rogue assembly (not a subassembly of our target) then we can follow that node and do the same thing. If both are rogue assemblies it does not matter which we follow. 

Since this is an assembly tree all the leaves represent single tile assemblies. Since we know none of the single tile assemblies are rogue assemblies (if it was the lemma would already be satisfied) we know at some point we must reach a node representing a rogue assembly that can be build from two subassemblies of our target $A$. 
\end{proof}

\section{Unique Assembly Verification Hardness} \label{sec:uavhard}
In this section, we show coNP-hardness of the Unique Assembly Verification problem in the 2HAM with constant temperature by a reduction from Monotone Planar 3-SAT with Neighboring Variable Pairs.

\begin{problem}[Monotone Planar 3-SAT with Neighboring Variable Pairs (MP-3SAT-NVP)]

\textbf{Input:} Boolean formula $\phi=C_1 \land \dots \land C_m$ in 3-CNF form where each clause only contains positive or negated literals from $X = \{x_1,\dots, x_n\}$. Further, any clause of $\phi$ with $3$ variables is of the form $(x_i, x_{i  + 1}, x_j)$ or $(\lnot x_i \lor \lnot x_{i + 1} \lor \lnot x_j)$, i.e., at least two of the literals are neighbors.
\textbf{Output:} Does there exist a satisfying assignment to $\phi$?
\end{problem}

Monotone Planar 3-SAT with Neighboring Variable Pairs was recently shown to be NP-hard in~\cite{Agarwal:2021:SODA}. We assume the instance of the problem is a rectilinear planar embedding where each variable is represented by a unit height rectangle arranged in the \emph{variable row}. Any planar 3SAT formula has a rectilinear encoding \cite{rectilinear}. We also assume that every clause is a unit-height rectangle with edges connecting the clauses and the contained variables. 
The monotone property ensures that each clause contains either only positive or only negative literals. Thus, the clauses may be separated with all positive clauses above the variable row, and all the negative clauses below. The final restriction is neighboring variable pairs, which states that for the three variables in each clause, at least two of the variables are neighbors in the variable row. An example instance is shown in Figure \ref{fig:SATfull}.

\subsection{Overview}

Given an instance of MP-3SAT-NVP $\phi$, we build an assembly $A$ and a 2HAM system $\Gamma$ that uniquely assembles $A$ if and only if $\phi$ does not have a satisfying assignment. An example instance is shown in Fig. \ref{fig:SATfull} and \ref{fig:target}. Alternatively, $\Gamma$ produces a rogue assembly if and only if there exists a satisfying assignment to $\phi$.

The ability to place all positive clauses above the variables and negative clauses below, along with the neighboring variable pairs, allows the clauses to be built hierarchically from the variables up. These properties allow us to require all nested clauses be evaluated and built before the outer clause is built. Thus, we define \emph{parent} and \emph{child} clauses as well as  \emph{root} clauses. In Figure \ref{fig:SATfull}, dotted lines connect child clauses $c_1$ and $c_2$ with their parent $c_3$. The root clauses are $c_3$ and $c_5$.\footnote{While a formula may have multiple clauses without a parent, the authors of \cite{Agarwal:2021:SODA} show that by adding additional variables, an instance may be constructed with only a single root clause. For MP-3SAT-NVP, we need at least two clauses (one for the positive and negative sides).}

\begin{definition}[Parent/Child/Root Clause]
Given a rectilinear encoding of Monotone Planer 3-SAT, a clause $C_p$ is a \emph{parent} clause of \emph{child} clause $C_c$, if $C_p$ fully encloses $C_c$, and any other clause that encloses $C_c$ also encloses $C_p$. 
A \emph{root} clause is a clause without a parent.
\end{definition}

Since $\phi$ is monotone, the positive and negative clauses may be separated across the variable row. The assembly $A$ is also separated by a horizontal bar that splits the assembly in two. This bar partially extends downward to prevent this assembly from attaching to itself. Above this bar is a subassembly that encodes the positive clauses and below the bar is a subassembly that encodes the negative clauses, which we call the positive and negative circuit, respectively.

The target assembly is designed so that it must be built from the variables up to the clauses. The clause gadgets can only be built if they are satisfied. Thus, parent clauses require that their variables or child clauses be satisfied to build the gadget. We will ensure this by using AND and OR gadgets between the variable and clause gadgets. 
We cover the parts of the system and gadgets in the order they must assemble:
 \begin{itemize} \setlength\itemsep{0em}
    \item Section \ref{subsec:vars}: variable gadgets  
    \item Section \ref{subsec:orclause}: OR gates and non-parent clause gadgets 
    \item Section \ref{subsec:andpar}: AND gates and parent clauses 
    \item Section \ref{subsec:rootarms}: the root clauses and the horizontal bar 
    \item Section \ref{subsec:rogue}: how a rogue assembly may form if and only if $\phi$ is satisfiable
 \end{itemize}

\begin{figure}[t]
	\centering
	\begin{subfigure}[b]{0.36\textwidth}
		\centering
		\includegraphics[width=1.\textwidth]{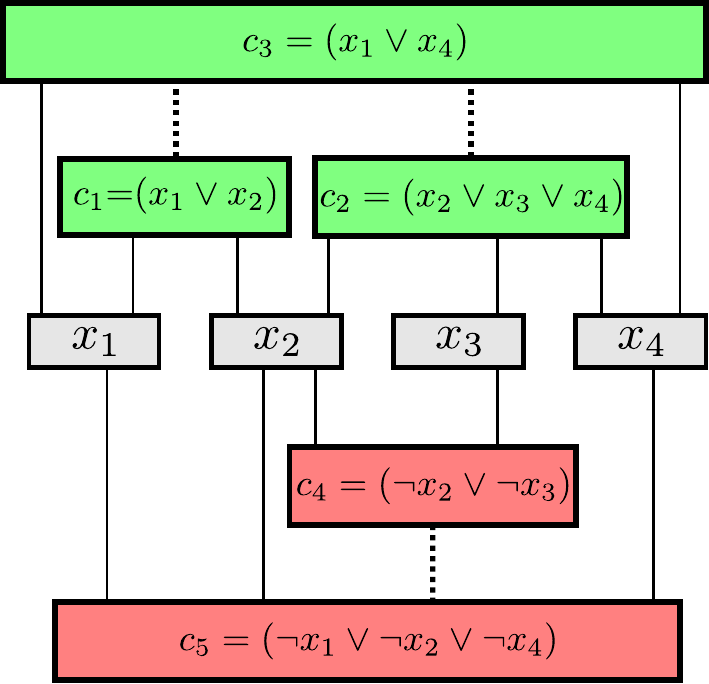}
		\caption{MP-3SAT-NVP Instance}
		\label{fig:SATfull}
	\end{subfigure}
	\begin{subfigure}[b]{0.3\textwidth}
		\centering
		\includegraphics[width=.8\textwidth]{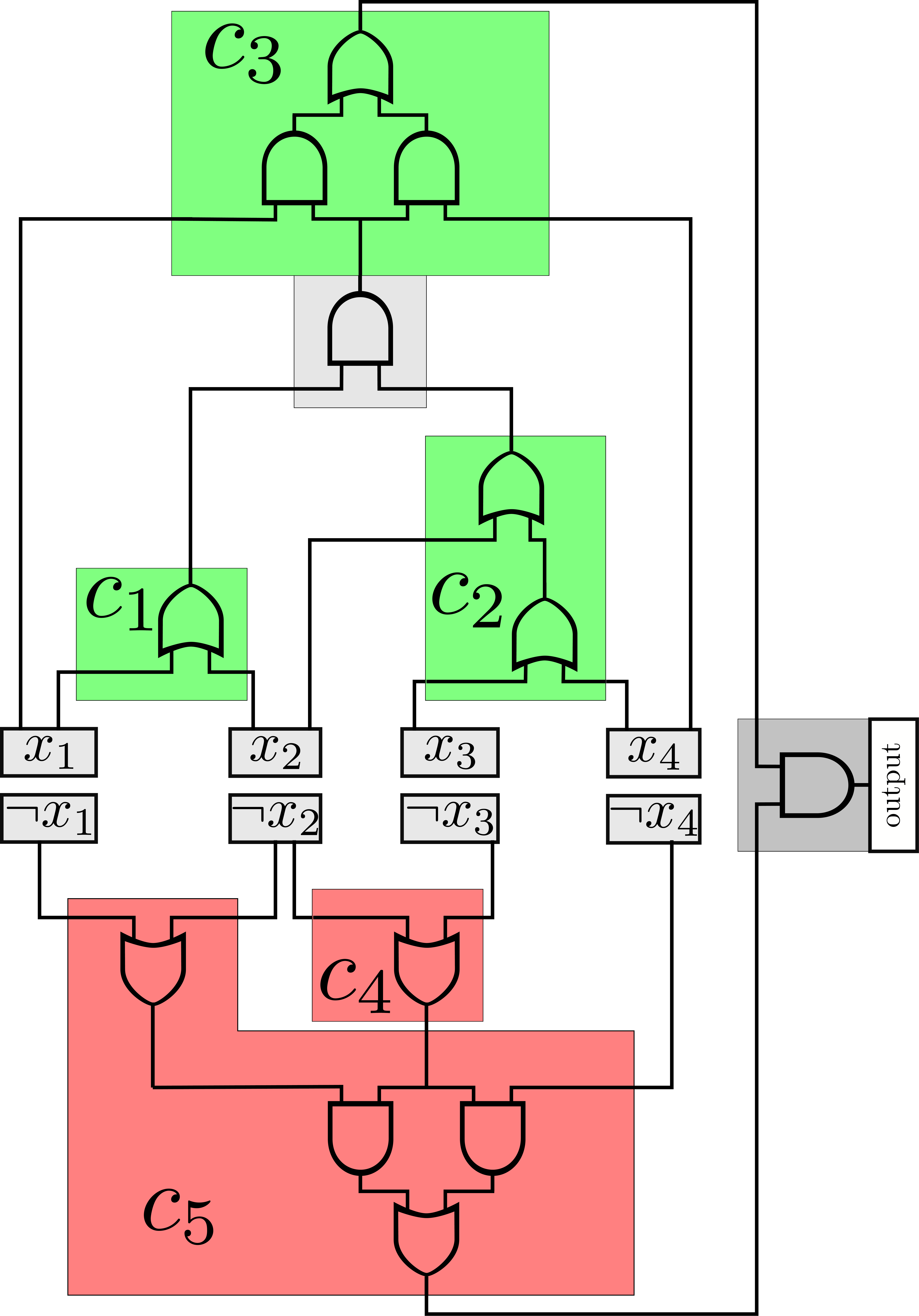}
		\caption{Target Assembly Circuit}
		\label{fig:targetcircuit}
	\end{subfigure}
	\begin{subfigure}[b]{0.3\textwidth}
		\centering
		\includegraphics[width=.85\textwidth]{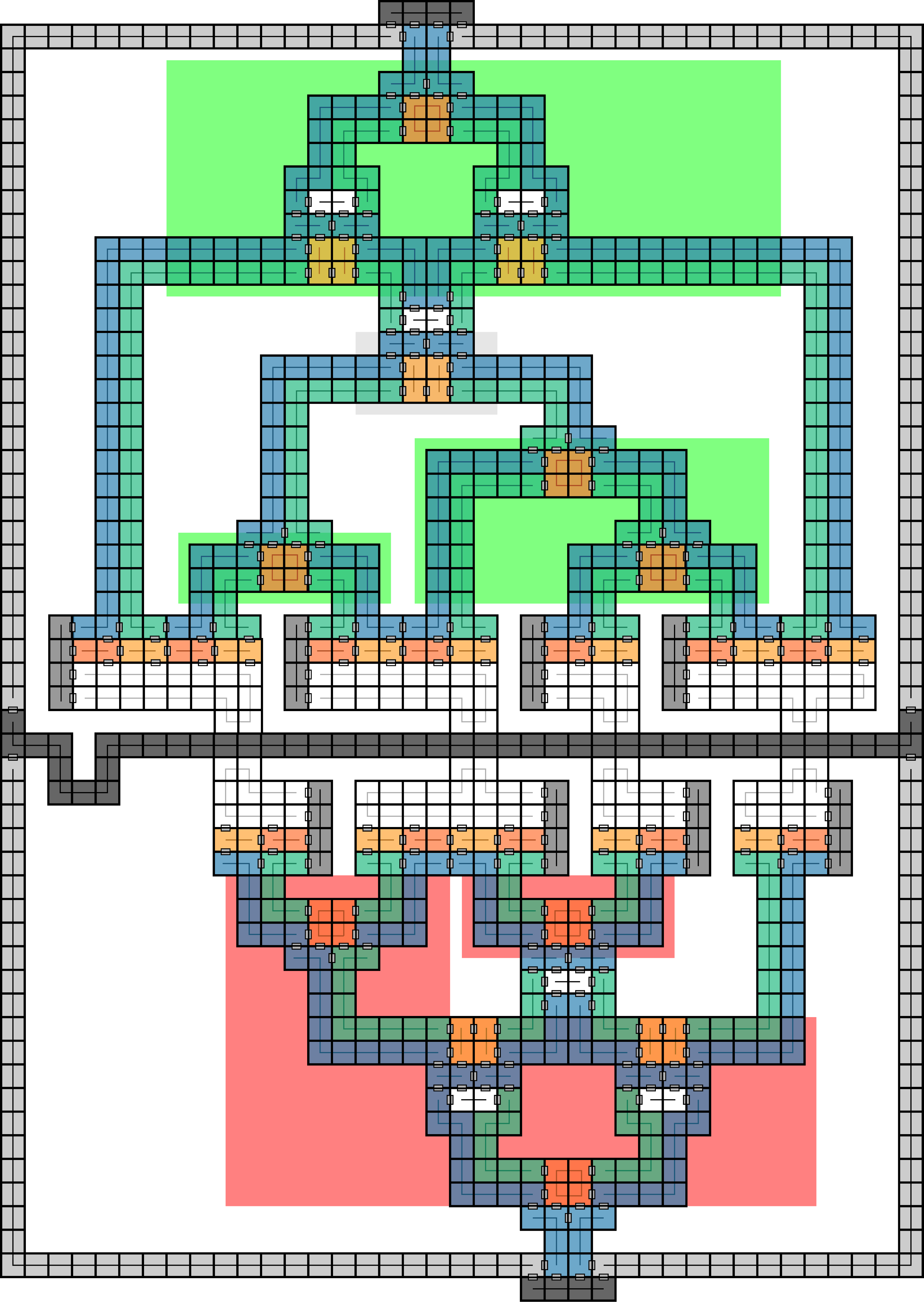}
		\caption{2HAM UAV Target Assembly}
		\label{fig:target}
	\end{subfigure}
	\caption{(a) Example instance of Monotone Rectilinear 3SAT with Neighboring Variable Pairs. Dotted lines are drawn between parent and child clauses. In this example $c_3$ and $c_5$ are the positive and negative root clauses respectively. 
		(b) A circuit view of our example instance with gates divided into the clauses they compute. We add AND gates (shown in grey) between child clauses that have the same variables. 
	(c) Target assembly constructed from instance on left. Each tile in the assembly is a unique tile type. 
	Each glue is unique except for the strength $1$ glues connecting the horizontal bar and the arms of each circuit. 
	The parts of the assembly that represent each clause are boxed in. 
 }
\end{figure}

\subsection{Variable Gadget} \label{subsec:vars}
For each variable gadget we use $(2 + 4d)$ subassemblies (Figure \ref{fig:parts}) to build the variable gadget where $d$ is the number of times the variable is used or it's outdegree. An example is shown in Figure \ref{fig:var1} where $d=1$ and Figure \ref{fig:var2} where $d=2$. In the figures, the lines are strength-2 glues, and the rectangle glues are all strength-1, thus requiring cooperative binding for the subassemblies to attach in a specific build order. We draw our gadgets separated into subassemblies but we construct our tile set using the single tiles which will self-assemble into these subassemblies. Every variable gadget is built as follows.

\begin{itemize}
    \item The \textbf{Bar Assembly} acts as a backbone (or separator) for the completed circuit subassemblies to connect to each other. 
    \item The \textbf{Bump} is the first assembly to attach to the Bar Assembly. The Bump is a height $2$ rectangle with an extra domino below it that is used for geometric blocking and encoding the assignment to that variable. The position of this domino is dependent on the position of the variable gadget on the opposite (negative) side.
    \item The \textbf{Base Dominoes} are used as part of the process of duplicating a variable path to multiple clauses. For each clause a variable is in, we use four subassemblies to connect to the next gadget. The first two gadgets are the Base dominoes. Once the Bump attaches to the Bar Assembly, the Base Dominoes can attach cooperatively to both. Once the first Base Domino attaches the next can attach using the glue from the previous domino and the other from the Bump.\footnote{Without these base dominoes, variable gadgets could build without the full Bump due to backfilling or backwards growth. The dominoes ensure this can not happen.}
    \item The \textbf{Wires} attach the variable gadgets to the clauses, and are the final two subassemblies for connecting to the next gadget. The first wire attaches cooperatively to the Bar Assembly and subsequent ones attach to the previous wire. The wires in our system are all built from two assemblies. When both halves of the wire are connected, the next gadget may attach. We call this a completed wire.
\end{itemize}

Variable gadgets in the negative circuit are built symmetrically rotated $180$ degrees. We adjust the position of dominoes on the Bumps of the gadgets so that variable gadgets on opposite circuits that represent the same variable have their domino in the same column.
We may generalize these gadgets to out degree $d$ (the variable appears in $d$ clauses) by increasing the width of the bump, and adding additional dominoes and wires.

\begin{figure}[t]
	\centering
	\begin{subfigure}[b]{0.23\textwidth}
		\centering
		\includegraphics[width=1.\textwidth]{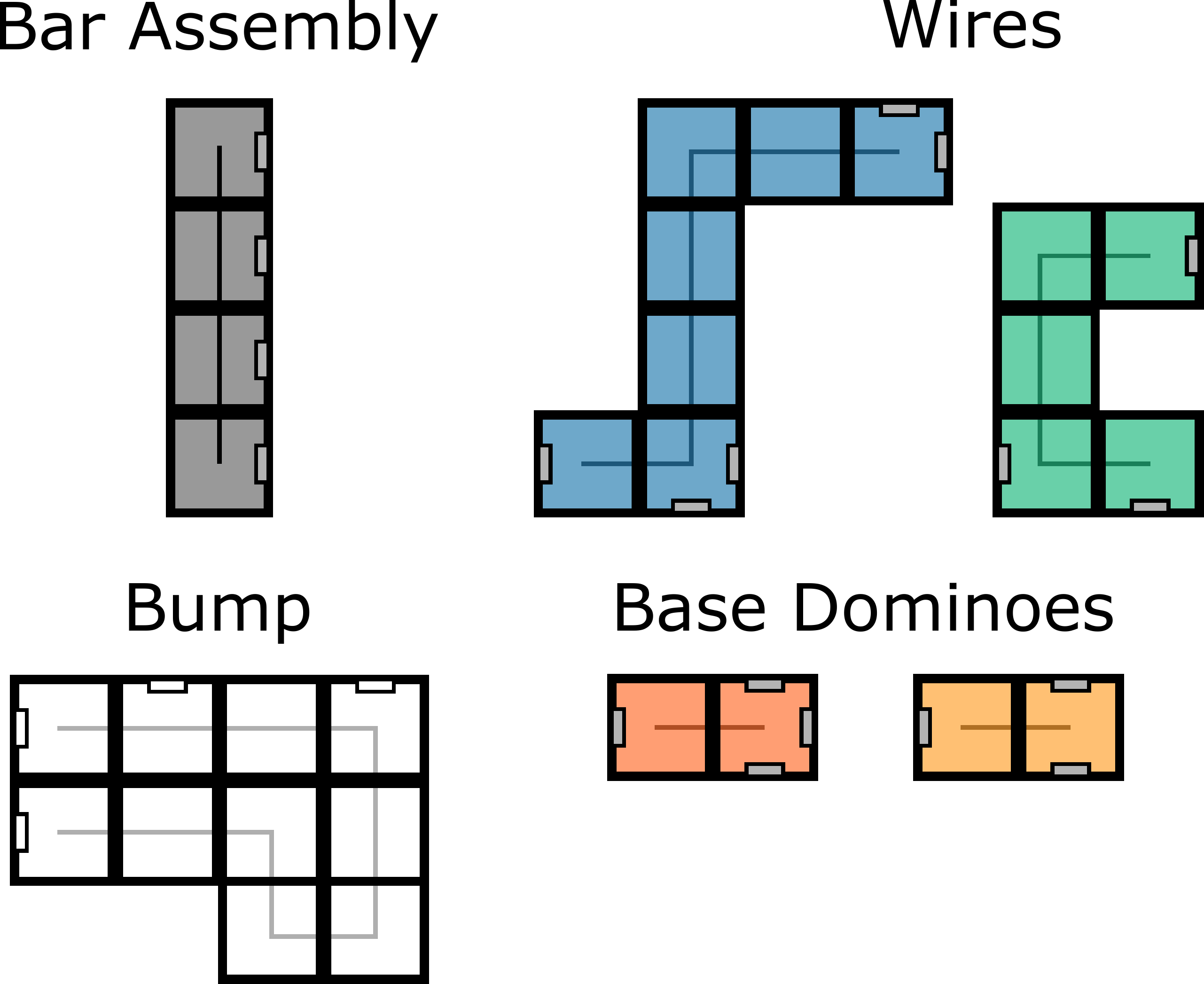}
		\caption{Subassemblies of a variable gadget}
		\label{fig:parts}
	\end{subfigure}
	\begin{subfigure}[b]{0.2\textwidth}
		\centering
		\includegraphics[width=.5\textwidth]{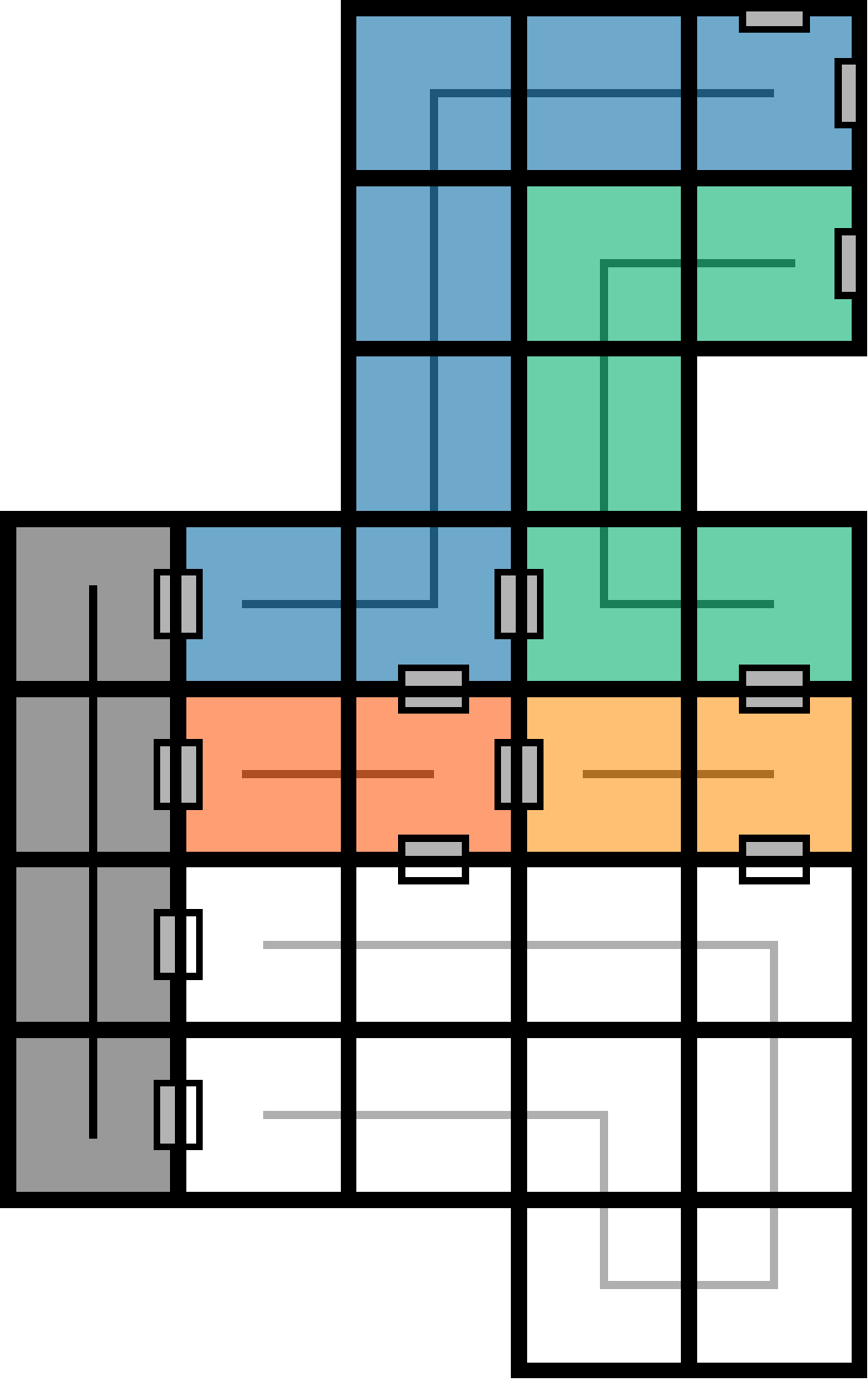}
		\caption{A variable gadget with outdegree $1$}
		\label{fig:var1}
	\end{subfigure}
	\begin{subfigure}[b]{0.23\textwidth}
		\centering
		\includegraphics[width=.65\textwidth]{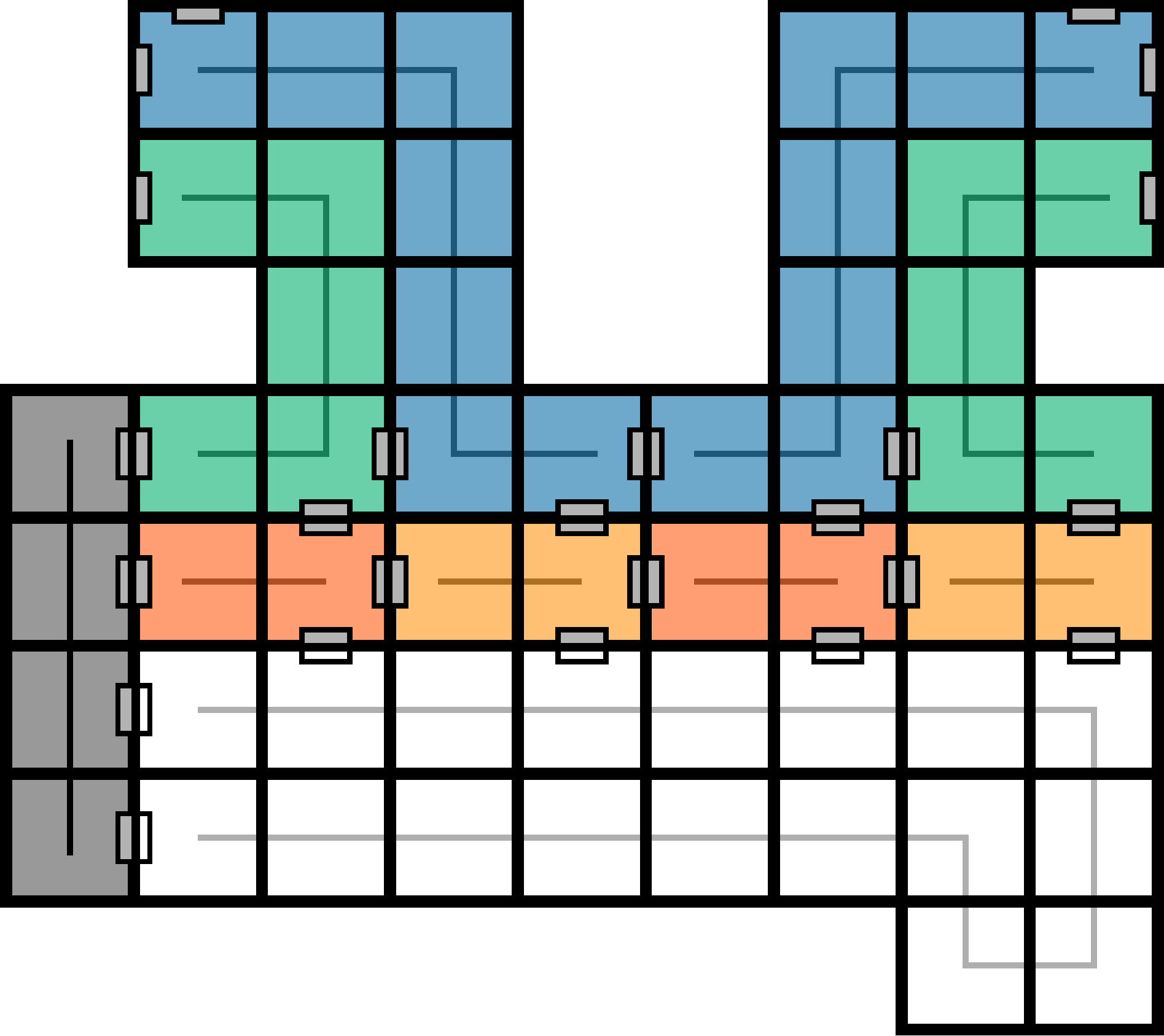}
		\caption{A variable gadget with outdegree $2$}
		\label{fig:var2}
	\end{subfigure}
	\begin{subfigure}[b]{0.25\textwidth}
		\centering
		\includegraphics[width=.9\textwidth]{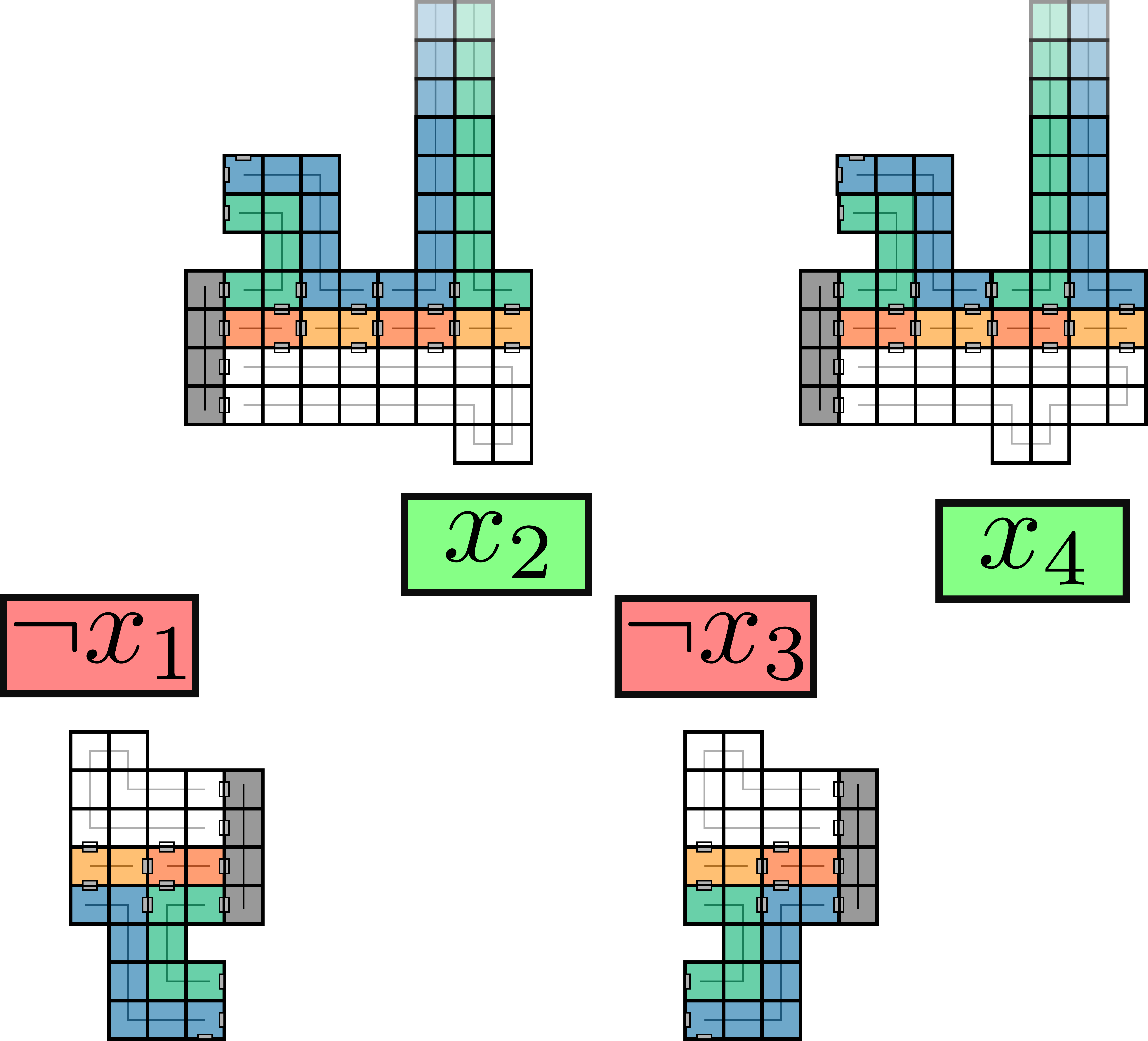}
		\caption{Assignment to the example MP-3SAT-NVP instance}
		\label{fig:assignment}
	\end{subfigure}
	\caption{(a) The smaller assemblies used for building a variable gadget. (b) A variable gadget for a variable that is used in a single clause. (c) A variable gadget with outdegree 2. For each additional output more base dominoes and wires are added. (d) A set of producible subassemblies representing variables that satisfy the example instance. We will walk through how these assemblies grow into a rogue assembly.}
\end{figure}

\subsection{OR Gates and Clause Gadgets} \label{subsec:orclause}

In CNF form, every variable in a clause is separated by a logical OR, thus, as part of our clause gadgets, we create OR gates to bring the variables together to ensure that the clause only forms if there is at least one variable that satisfies the clause.

\begin{figure}[t]
	\centering
	\begin{subfigure}[b]{0.46\textwidth}
		\centering
		\includegraphics[width=.9\textwidth]{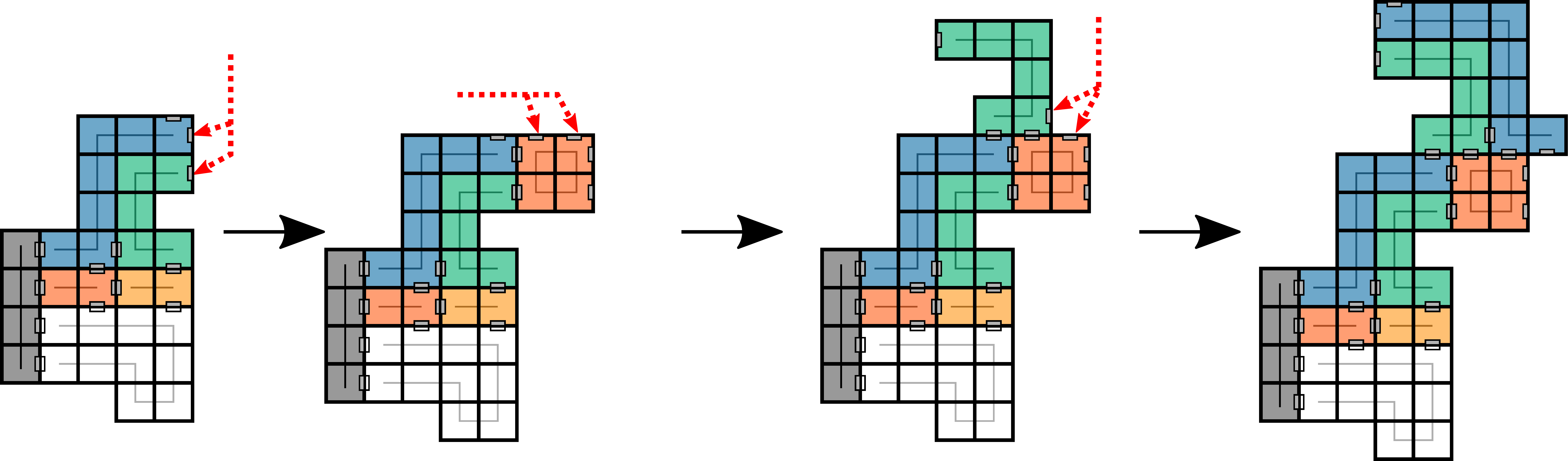}
		\caption{Build OR gate}
		\label{fig:buildO}
	\end{subfigure}
	\hspace*{.5cm}
	\begin{subfigure}[b]{0.46\textwidth}
		\centering
		\includegraphics[width=.9\textwidth]{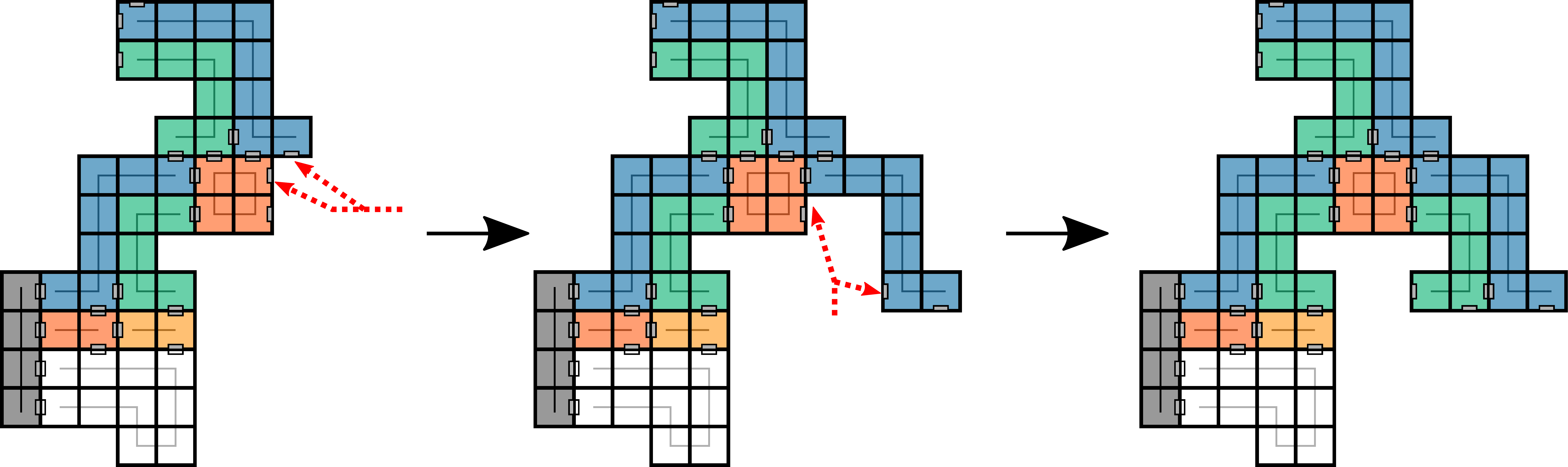}
		\caption{Wire backing filling}
		\label{fig:backFilling}
	\end{subfigure}
	\caption{(a)  The process of a variable gadget growing the OR gate used for clauses. Glues used for attachment in the next step are denoted by arrows. If one of the variable gadgets is constructed the $2 \times 2$ square assembly may attach. The output wires of the gate then attach cooperative with the wire from the variable assembly and the square. Note that only one of the variable assemblies needs to be constructed for the OR gate to build it's output wire. (b) Once the output wires of the OR gate have attached the wire for the other variable may ``Backfill'' or grow backwards.}
\end{figure}

\paragraph{OR Gates.}
An example of how the OR gate grows off of a variable gadget is shown in Figure \ref{fig:buildO}.
The OR gate consists of a single $2 \times 2$ square with strength-$1$ glues on the west, north, and east facing tile edges. The west and east glues each connect to the wires that input to the gate. The north glues are used cooperatively with glues on the incoming wire gadgets to attach another wire gadget going to the clause gadget. To complete the new wire gadget it must also cooperatively use the other incoming wire.
We note that the wires from the other input can backfill, but that does not cause an issue as the `backward' growth stops after building the wire. Figure \ref{fig:backFilling} shows an example with only one variable used in the OR gate.

\begin{figure}[t]
	\centering
	\begin{subfigure}[b]{0.3\textwidth}
		\centering
		\includegraphics[width=1.\textwidth]{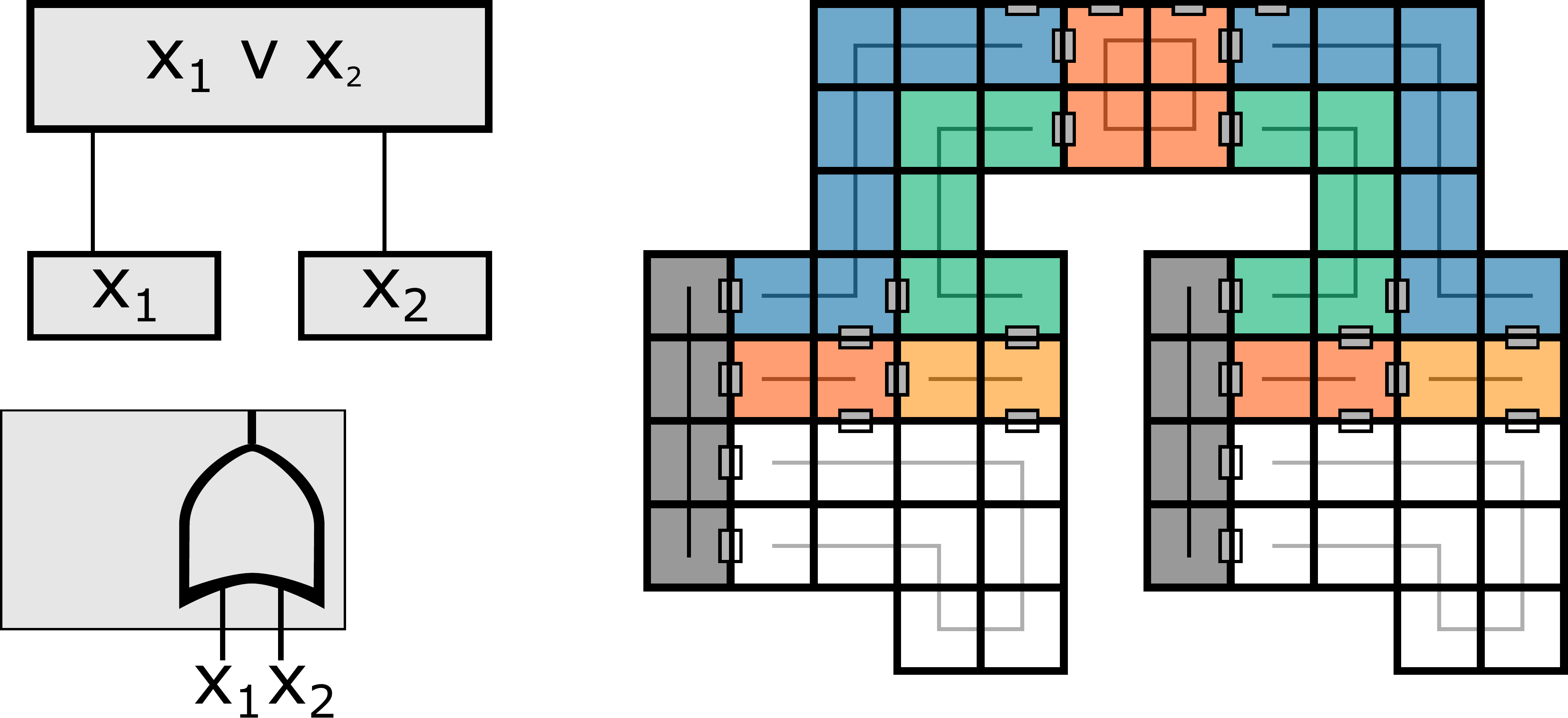}
		\caption{Clause Gadget}
		\label{fig:neighClause}
	\end{subfigure}
	\begin{subfigure}[b]{0.3\textwidth}
		\centering
		\includegraphics[width=1.\textwidth]{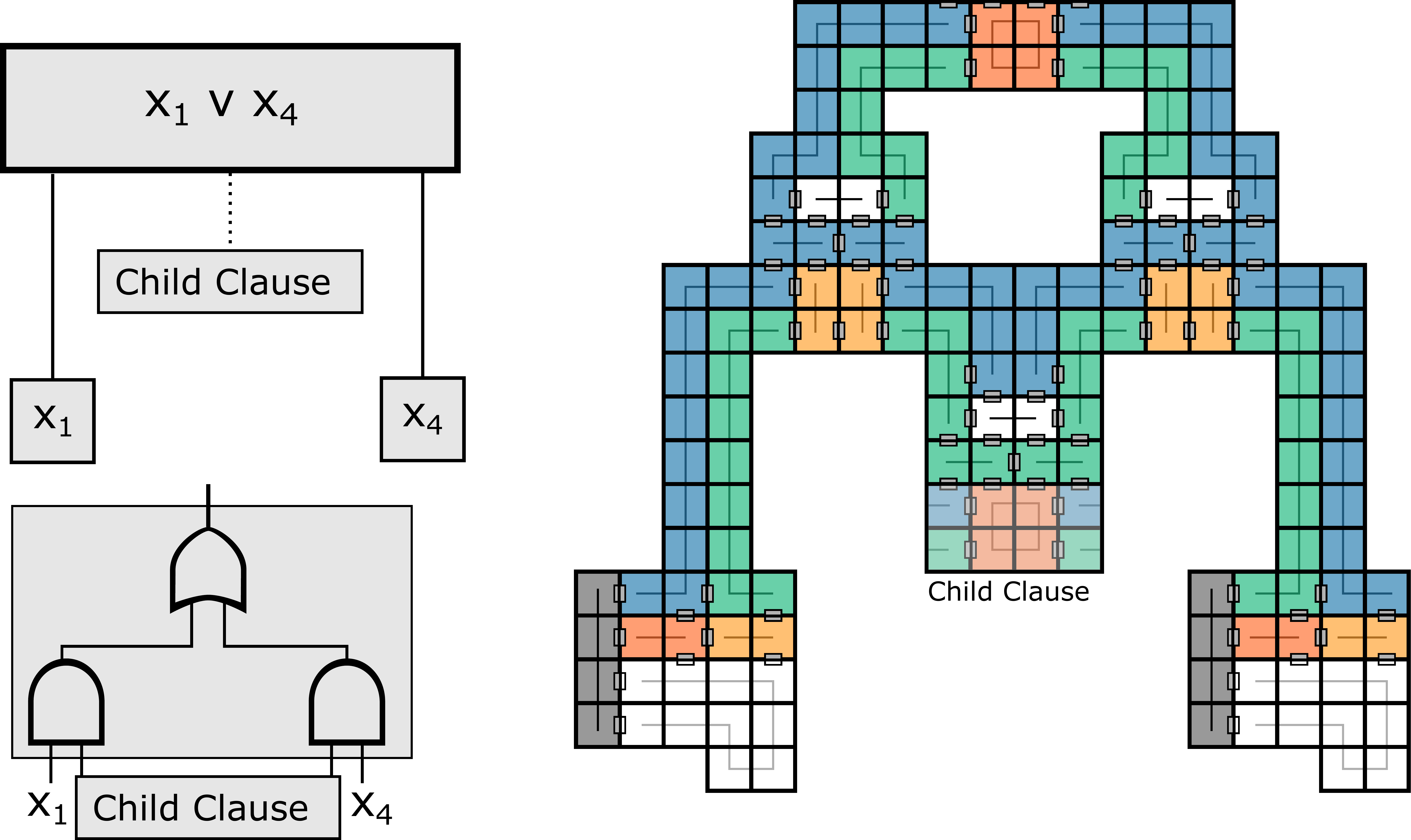}
		\caption{Parent Clause gadget}
		\label{fig:parent}
	\end{subfigure}
	\begin{subfigure}[b]{0.3\textwidth}
		\centering
		\includegraphics[width=.85\textwidth]{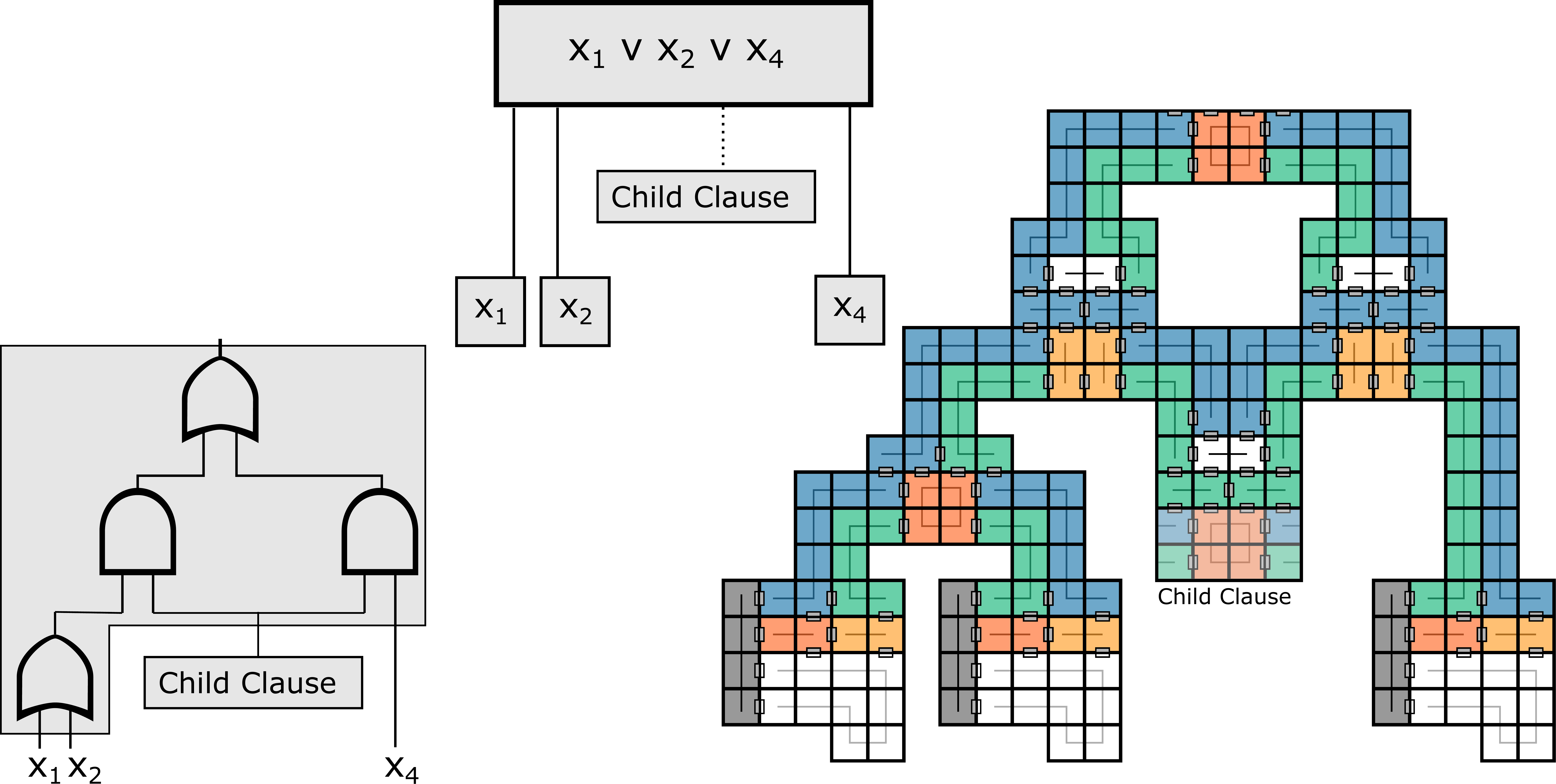}
		\caption{Parent clause gadget with 3 literals}
		\label{fig:3clause}
	\end{subfigure}
	\caption{(a) A clause gadget with $2$ neighboring variables. (b) A clause gadget with two variables and a child clause. (c) When a parent clause has $3$ literals we know two of them must be neighbors. Using an additional OR gate we may  use the same gadget as the clause with $2$ literals. }
\end{figure}

\begin{figure}[t]
	\centering
	\begin{subfigure}[b]{0.28\textwidth}
		\centering
		\includegraphics[width=.8\textwidth]{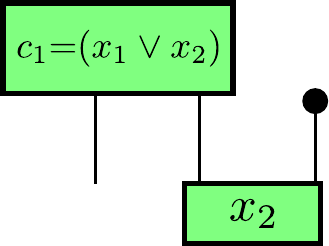}
		\caption{Clause $c_1$}
		\label{fig:clause1}
	\end{subfigure}
	\begin{subfigure}[b]{0.28\textwidth}
		\centering
		\includegraphics[width=.7\textwidth]{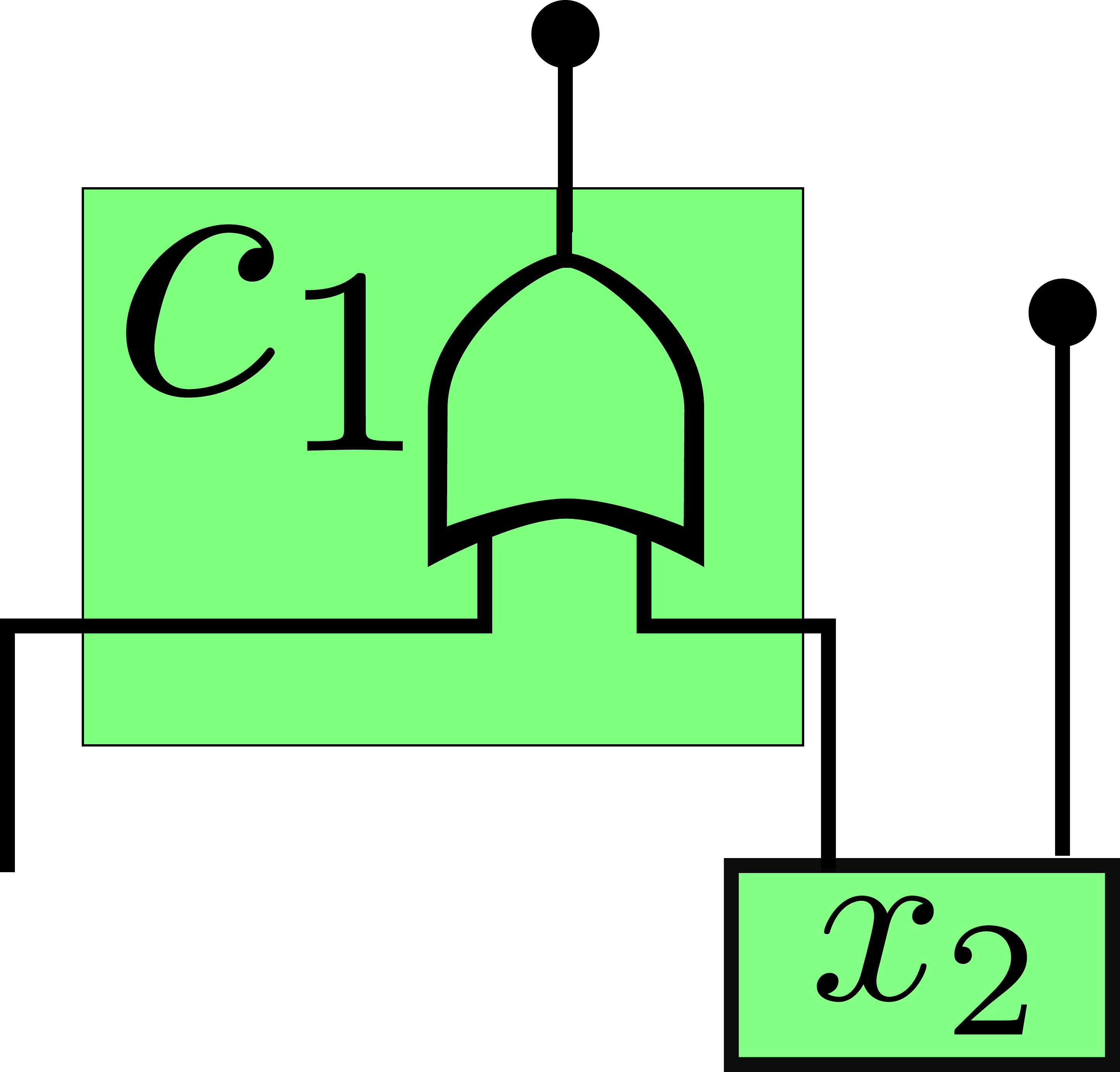}
		\caption{Circuit for $c_1$}
		\label{fig:neighClause}
	\end{subfigure}
	\begin{subfigure}[b]{0.28\textwidth}
		\centering
		\includegraphics[width=.8\textwidth]{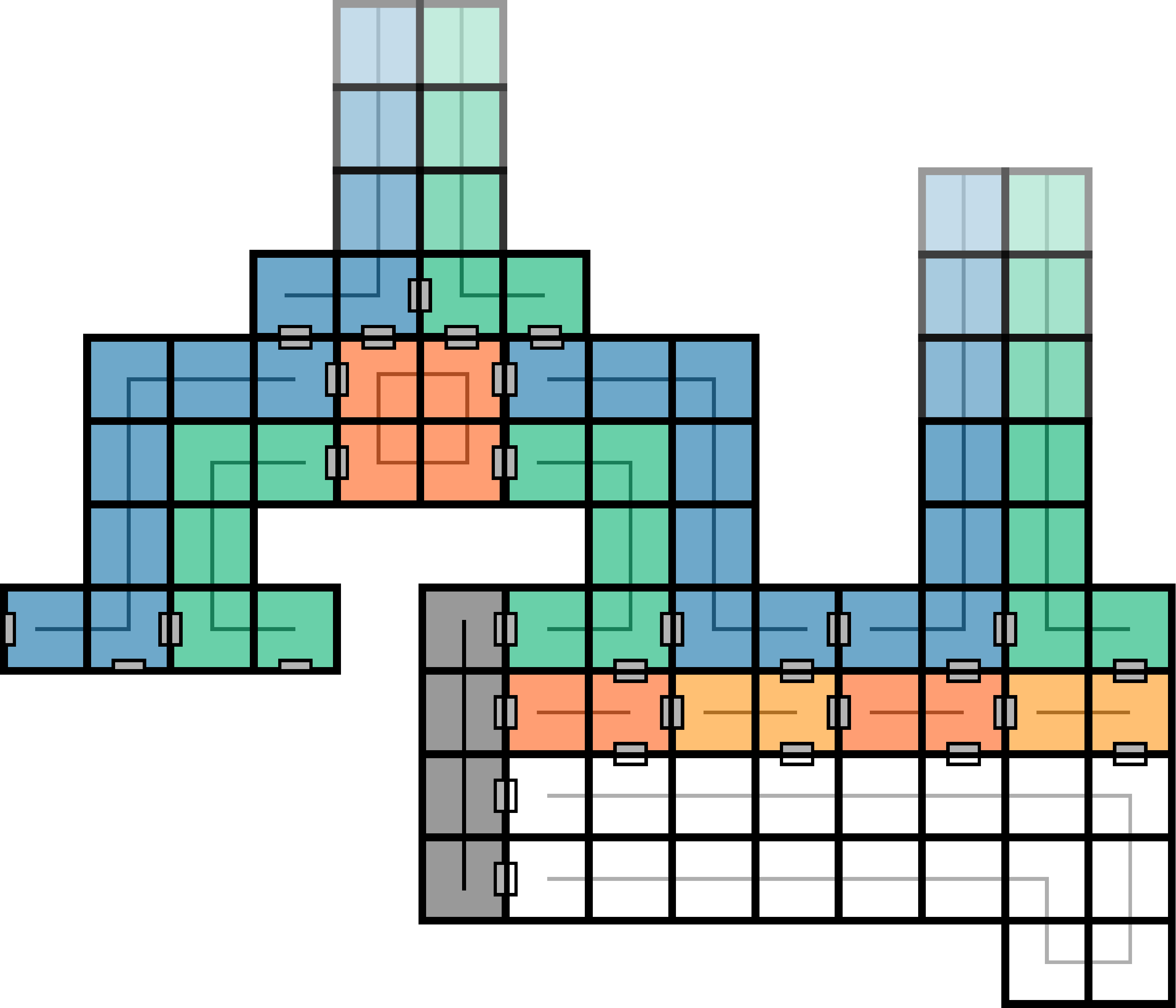}
		\caption{Gadgets for $c_1$}
		\label{fig:clause1g}
	\end{subfigure}
	\caption{(a)  The clause $c_1$ in Figure \ref{fig:SATfull} is satisfied by $x_2 = 1$. (b) The OR gate grows off of $x_2$. The other wires on the variable gadget are used to connect to other clause gadgets. (c) The gadget constructed for the clause $c_1$. Note the other wire from the OR gate has backfilled. }
\end{figure}

\begin{figure}[t]
	\centering
	\begin{subfigure}[b]{0.3\textwidth}
		\centering
		\includegraphics[width=.75\textwidth]{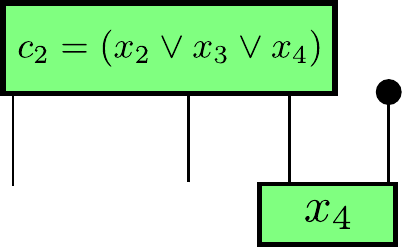}
		\caption{Clause $c_2$}
		\label{fig:clause2}
	\end{subfigure}
	\begin{subfigure}[b]{0.3\textwidth}
		\centering
		\includegraphics[width=.5\textwidth]{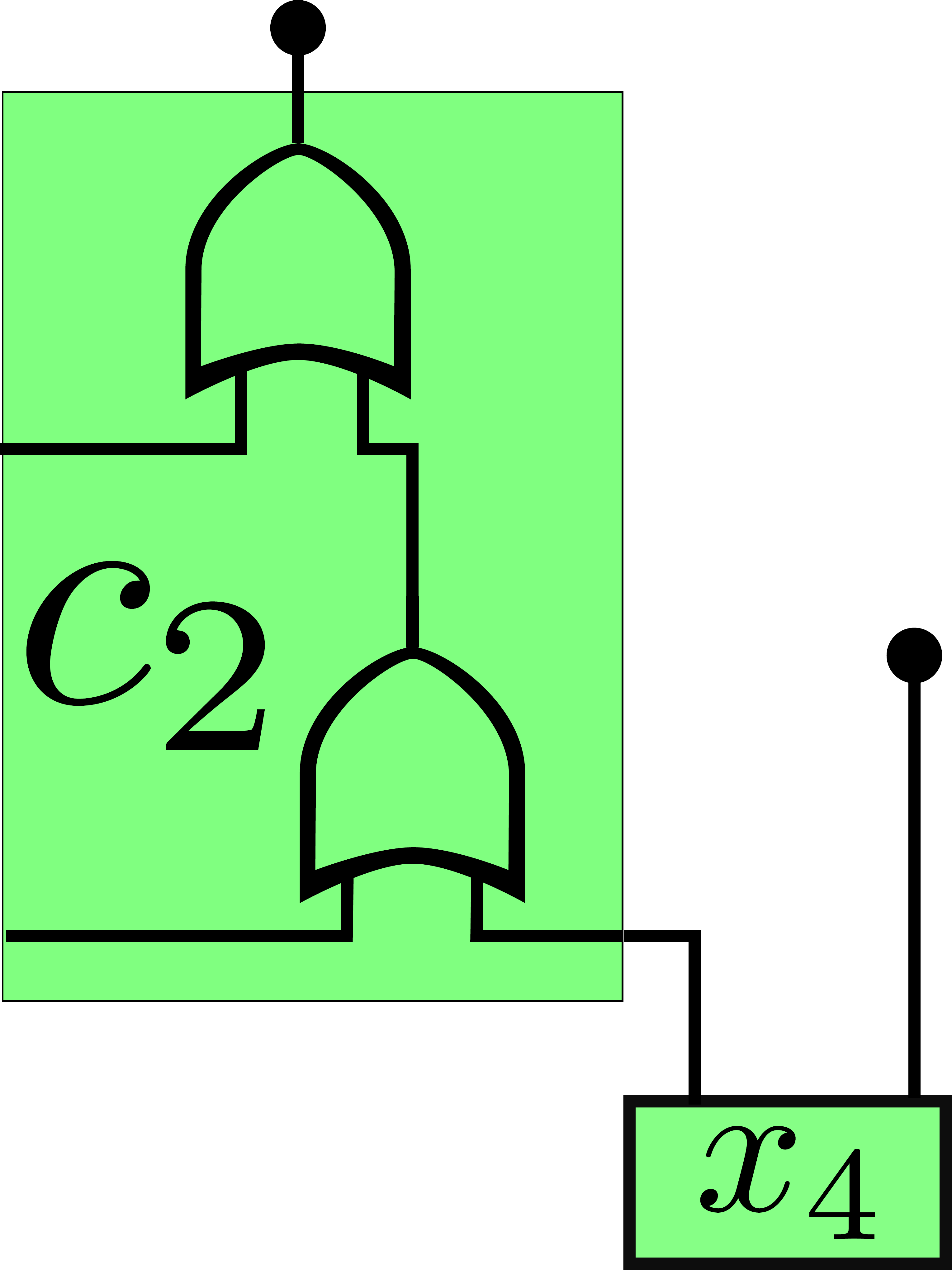}
		\caption{Circuit for $c_2$}
		\label{fig:clause2circuit}
	\end{subfigure}
	\begin{subfigure}[b]{0.3\textwidth}
		\centering
		\includegraphics[width=.75\textwidth]{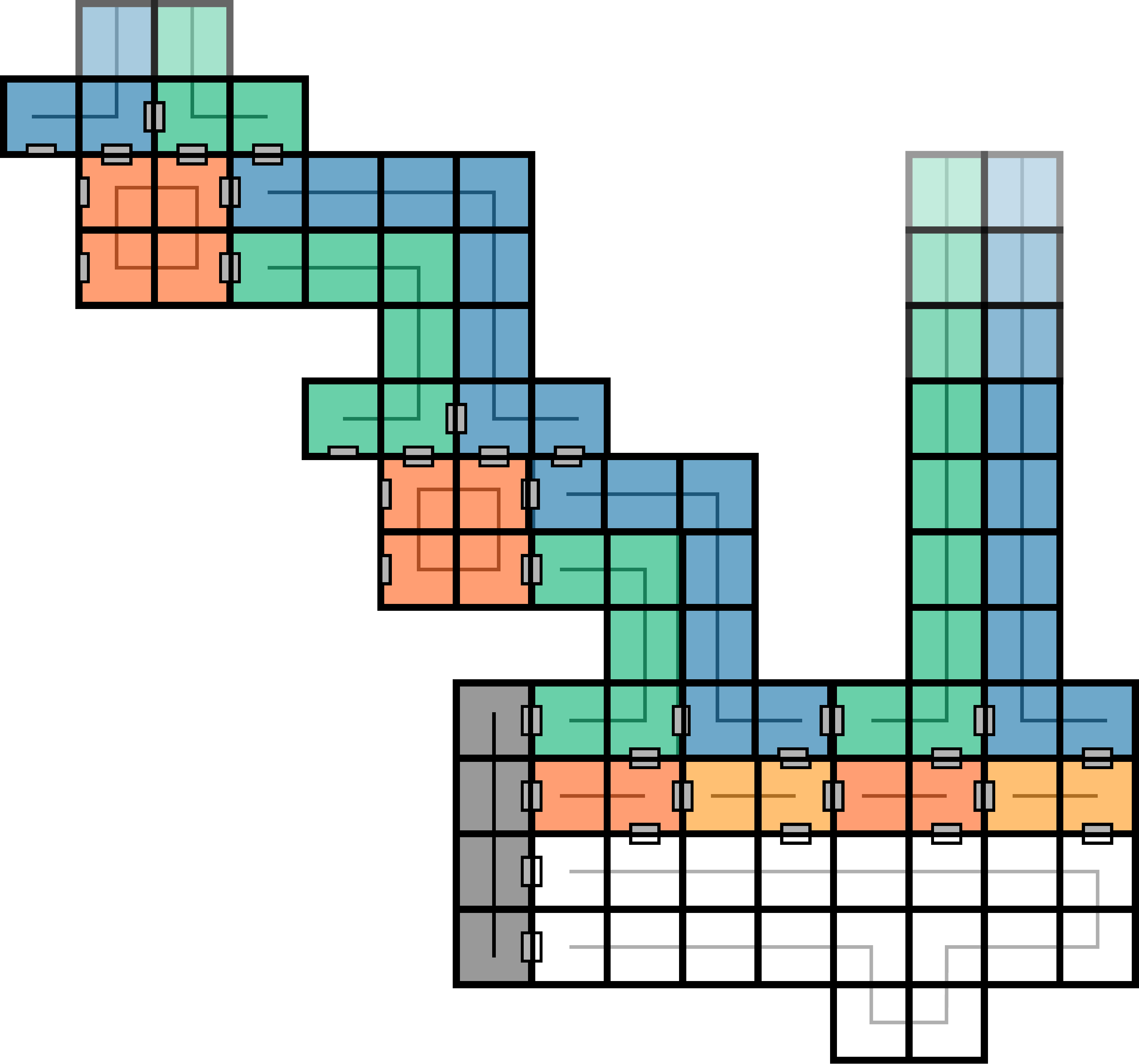}
		\caption{Gadgets $c_3$}
		\label{fig:clause2g}
	\end{subfigure}
	\caption{(a) $c_2$ from our example. This clause has 3 variables and no children. (b) The clause is computed using two OR gates. The gates are able to grow from $x_4$. (c) $x_4$ variable gadget allows for the two OR gates to attach. 
}
\end{figure}

\paragraph{Non-parent Clauses.} The first clause type we cover are clauses without children, or clauses at the bottom of the circuit. The simplest type of this gadget are clauses with only $2$ literals as in Figure \ref{fig:neighClause}. This gadget is fairly straightforward to implement as we only need to use a single OR gate. An example of this type of clause is in Figure \ref{fig:clause1}, and its implementation is in Figure \ref{fig:clause1g}. Note that both variables appear in other clauses so those variable gadgets have additional wires. 
For non-parent clauses with $3$ literals (Figure \ref{fig:clause3}), we use $2$ OR gates (Figure \ref{fig:clause2g}).

\subsection{AND Gates and Parent Clauses} \label{subsec:andpar}
Since every clause in CNF form is separated by a logical AND, we create AND gates that compare clauses. Thus, we need to know which clauses are parent clauses since they have child clauses underneath them with wires coming into the gates. We also build a FANOUT gate for connecting clauses.
 
\paragraph{AND Gates.} The AND gate uses $2$ vertical dominoes that share a single strength-$1$ glue between them. Figure \ref{fig:buildA} shows an example AND gate being constructed. Once a wire that inputs to the gate is completed, one of the dominoes can cooperatively attach. The domino has another strength-$1$ glue on its north side that allows a horizontal domino to cooperatively attach using the glue exposed on the wire. 

Using the glues from the newly attached dominoes, the two halves of the gate are able to attach to each other. This allows for the two glue on the horizontal dominoes to be used to cooperative bind the white center domino. From here, the two halves of the wire that outputs from the AND gate can attach.

\paragraph{FANOUT Gates.}
In order to build the parent clause, we also need a way to `fan-out' and copy the signal from an AND gate to two other gadgets. We do this by adding glues to the north side of the center domino and having two wires grow off of the gadget. This process is shown in Figure \ref{fig:buildFanOut}.

\begin{figure}[t]
	\centering
	\begin{subfigure}[b]{0.6\textwidth}
		\begin{subfigure}[b]{0.75\textwidth}
			\centering
			\includegraphics[width=.85\textwidth]{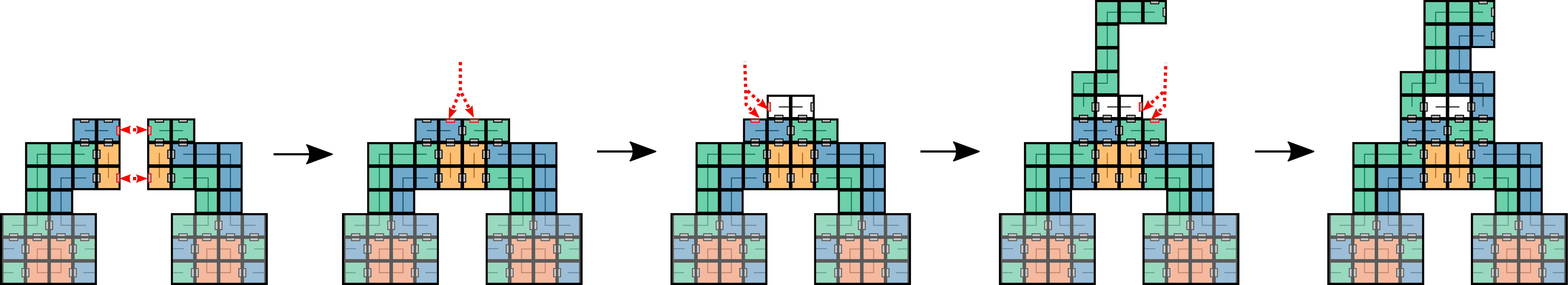}
			\caption{Build AND gate}
			\label{fig:buildA}
		\end{subfigure}
		\begin{subfigure}[b]{0.75\textwidth}
			\centering
			\includegraphics[width=.9\textwidth]{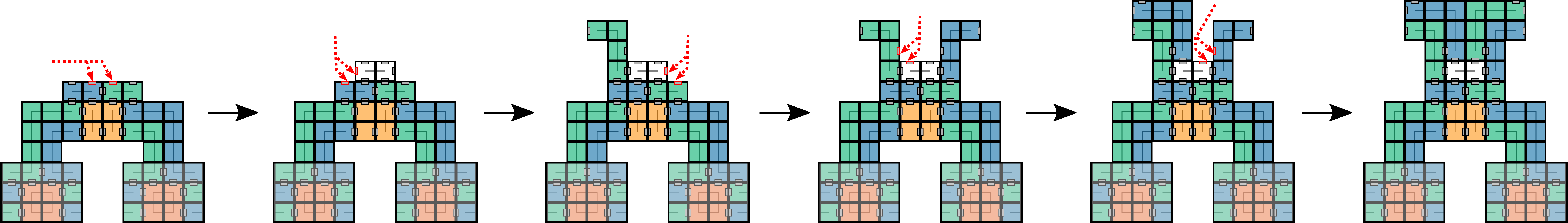}
			\caption{Build FANOUT}
			\label{fig:buildFanOut}
		\end{subfigure}
	\end{subfigure}
	\begin{subfigure}[b]{0.37\textwidth}
		\centering
		\includegraphics[width=.78\textwidth]{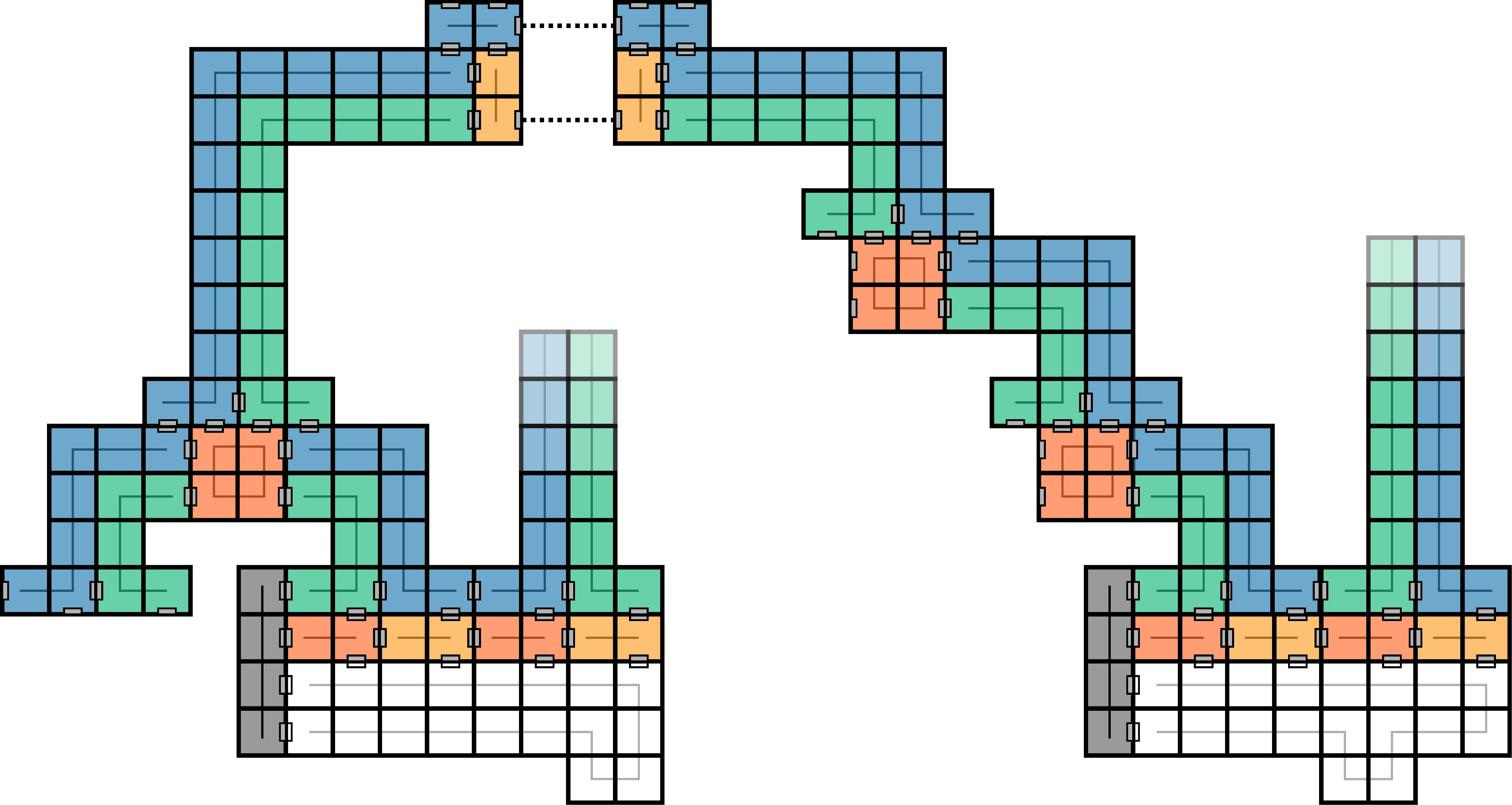}
		\caption{Attachment of an AND gate}
		\label{fig:childAND}
	\end{subfigure}
	\caption{(a) The process of an AND gadget assembling. The output wires can only grow from the combined halves of the AND gate.
	(b) By modifying the center domino, two wires may be output from a single AND gate, which works as a FANOUT. (c)  The two clauses $c_1$ and $c_2$ both have the same parent clause so they are joined by an AND gate. Once the dominoes attach to the output wire of the clauses the two assemblies may attach to each other.}
\end{figure}

\paragraph{Parent Clauses.}
Consider a parent clause $C_p = (x_1 \lor x_4)$. Let $C_c$ be the child clause. Since we want this gadget to build only if its own clause and its child are both satisfied, we can view this statement as $(x_1 \lor x_4) \land C_c$. However, we can modify the statement to be $(x_1 \land C_c) \lor (C_c \land X_4)$, which we can build since we have planar circuits. An example of the circuit and gadgets are shown in Figure \ref{fig:parent}.

\begin{figure}[t]
	\centering
	\begin{subfigure}[b]{0.28\textwidth}
		\centering
		\includegraphics[width=.85\textwidth]{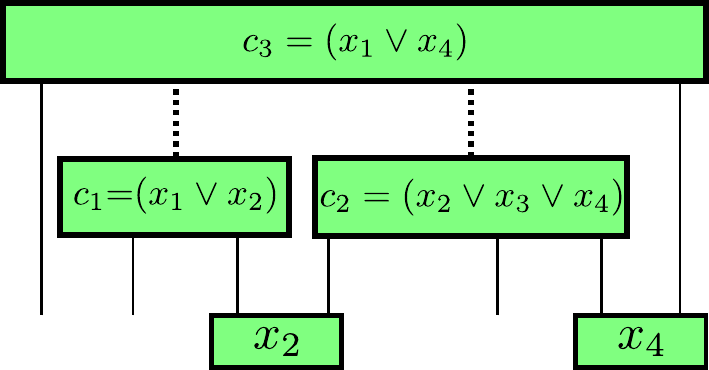}
		\caption{Clause $c_3$}
		\label{fig:clause3}
	\end{subfigure}
	\begin{subfigure}[b]{0.28\textwidth}
		\centering
		\includegraphics[width=.85\textwidth]{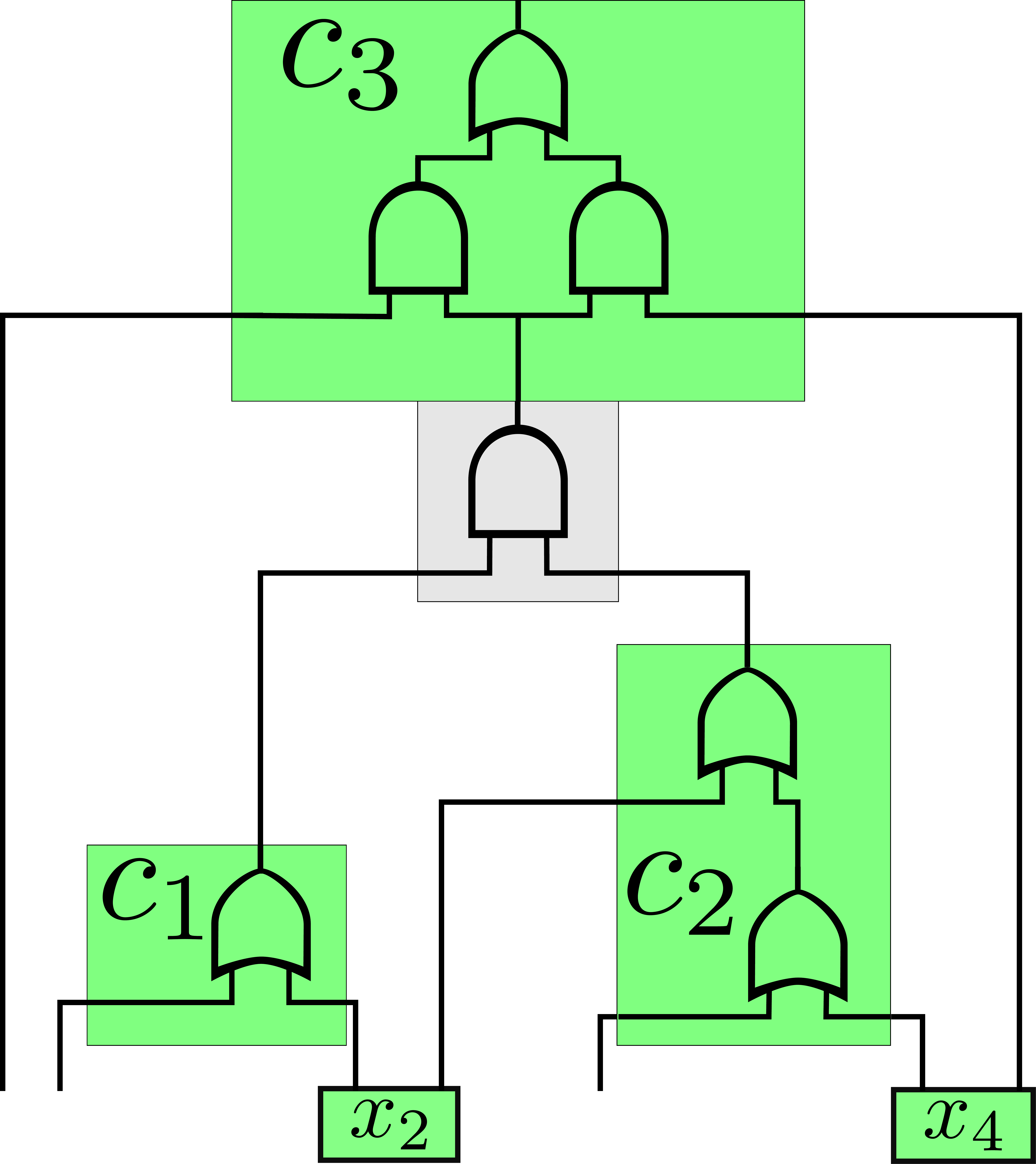}
		\caption{Circuit for $c_3$}
		\label{fig:clause3Cir}
	\end{subfigure}
	\begin{subfigure}[b]{0.28\textwidth}
		\centering
		\includegraphics[width=.78\textwidth]{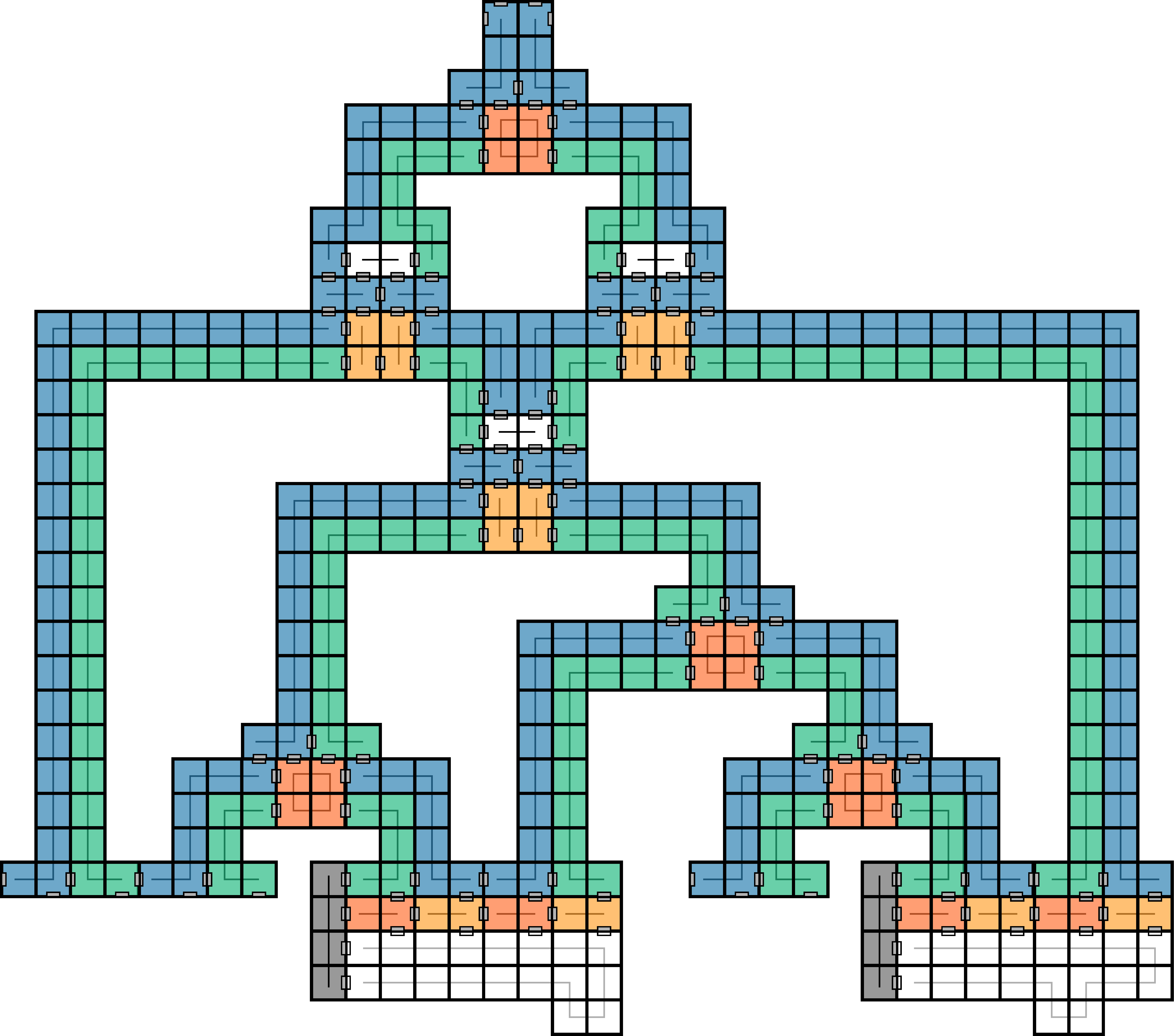}
		\caption{Gadget for $c_3$}
		\label{fig:clause3g}
	\end{subfigure}
	\caption{
	(a) The root clause of the example instance. This clause has two literals and two children.
	(b) Since the AND gate has built and $x_4$ satisfies $c_3$ the clause may grow. 
	(c) The two child clause's output are connected by an AND gate and then used as the middle input to the gadget. 
}
\end{figure}

By the neighboring variable pairs restriction, we know that any clause with three variables has at least a pair of them being neighbors. This means that there cannot be any child clauses beneath that neighboring pair, so we may use an OR gate between those two variables and then build the rest of the gadget in the same way as the two literal version (Figure \ref{fig:3clause}).

In our example instance, the root clause of the positive circuit has two children. For these cases we may use the AND gadget to verify that both child clauses have been satisfied before allowing the parent clause to build.
The root clause of the negative circuit in our example instance (Figure \ref{fig:clause5}) has three literals. The constructed gadget can be seen in Figure \ref{fig:clause5}.

\begin{figure}[t]
	\centering

	\begin{subfigure}[b]{0.28\textwidth}
		\centering
		\includegraphics[width=.7\textwidth]{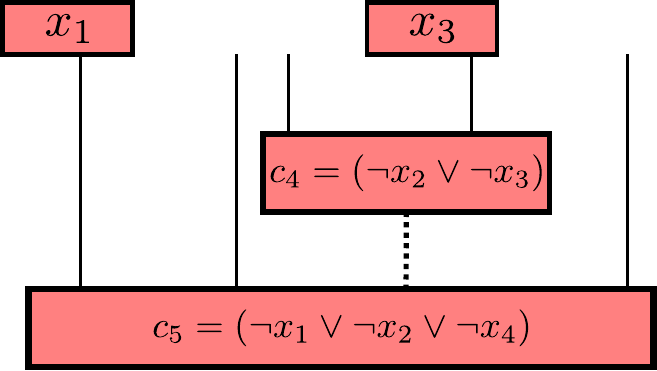}
		\caption{Clause $c_5$}
		\label{fig:clause5}
	\end{subfigure}
	\begin{subfigure}[b]{0.28\textwidth}
		\centering
		\includegraphics[width=.8\textwidth]{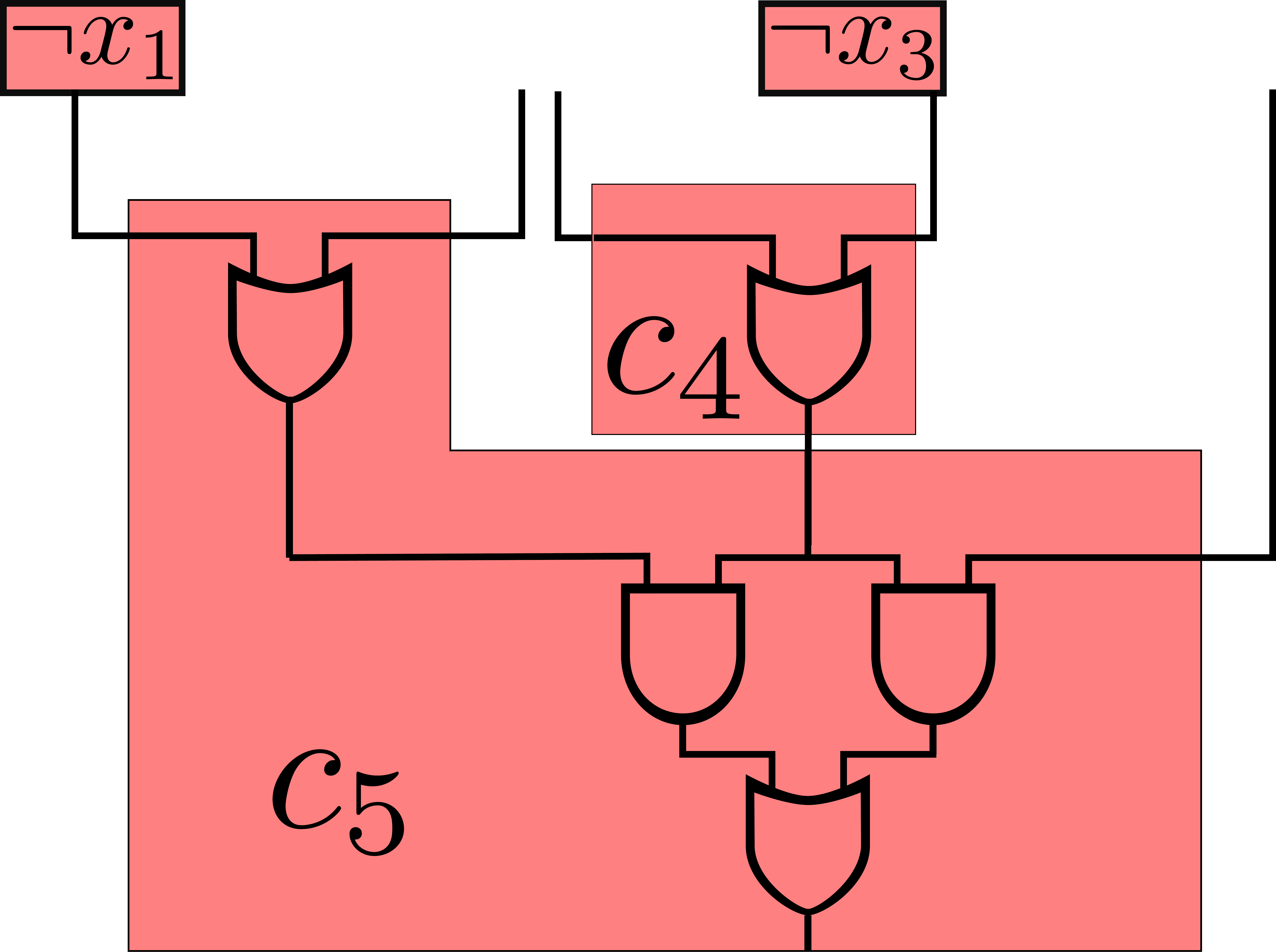}
		\caption{Circuit for $c_5$}
		\label{fig:clause5Cir}
	\end{subfigure}
	\begin{subfigure}[b]{0.28\textwidth}
		\centering
		\includegraphics[width=.8\textwidth]{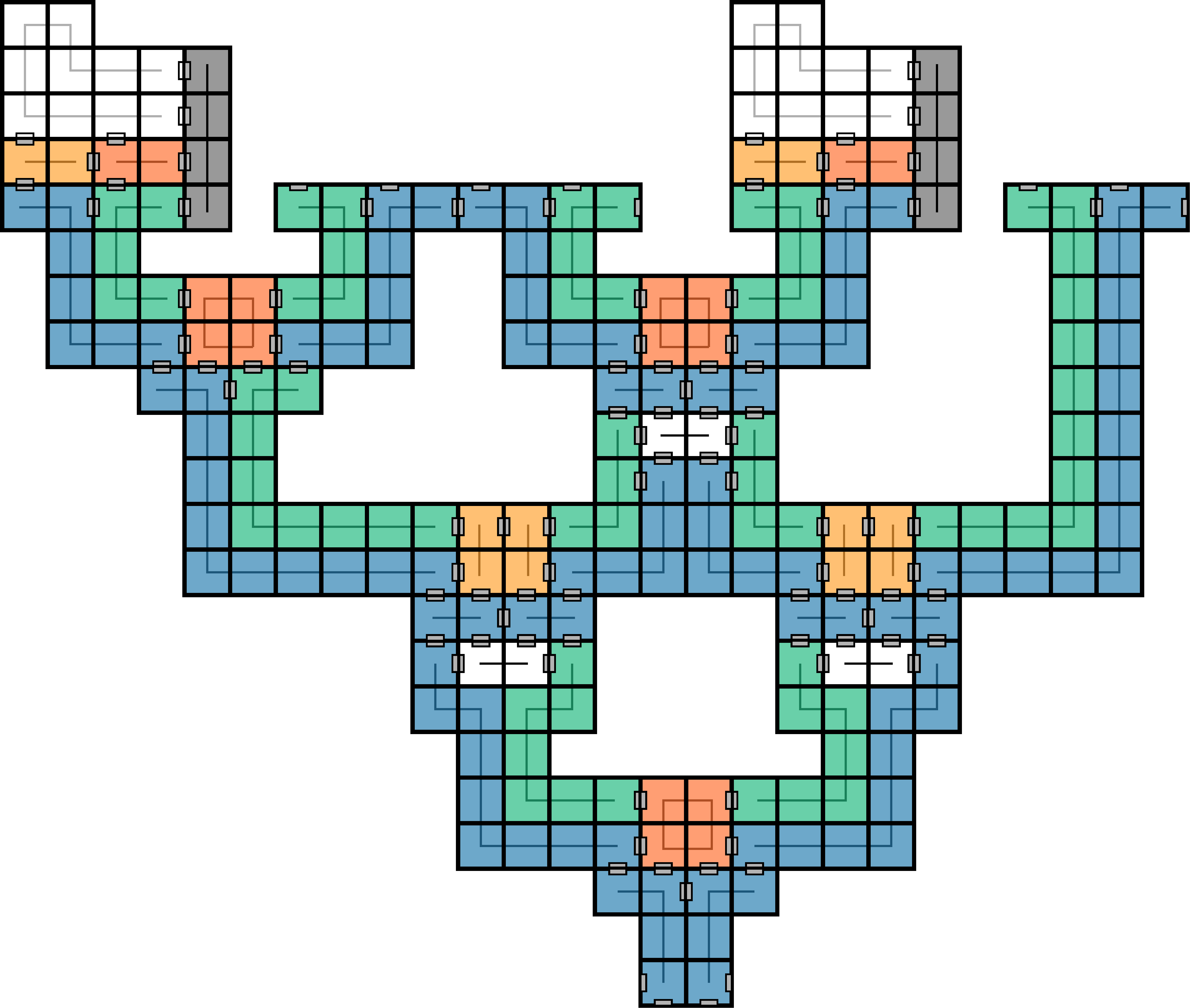}
		\caption{Gadget for $c_5$}
		\label{fig:clause5g}
	\end{subfigure}
	\caption{(a)  In the example instance the negative circuit has $c_5$  which is a parent clause with $3$ literals. 
	(b) The negative circuit draw with gates. The variables $x_1$ and $x_3$ being false satisfies all the clauses.
	(c) The variable assemblies we selected at the beginning also grow into a circuit with the root clause built.}
\end{figure}

\subsection{Root Clauses and Horizontal Bar} \label{subsec:rootarms}

\paragraph{Root Clauses and Arms.}
The root clause is the outermost clause on either side of the variables. Although it functions similar to the other clauses, instead of outputting a wire, a horizontal $4 \times 1$ rectangle can attach after it finishes assembling. The arms may then cooperatively bind to the rectangle and the wires of the root clause forming the top of the circuit. The glues on the ends of these arms allow for the circuit to attach to the horizontal bar. A high-level view of the root clauses and arms attached is shown in Figure \ref{fig:bar} as well as a detail of the assembly process of the root clause in Figure \ref{fig:rootTerm}.

\paragraph{Horizontal Bar.}
The horizontal bar (Figure \ref{fig:bar}) is a width-$1$ assembly that extends the width of both circuits with strength-$1$ glues on the north and south side of the outer tiles. Since the arms must also be able to attach to each other to form a rogue assembly the glues on the ends of the horizontal bar must be the same. In order to prevent the horizontal bar from attaching to another instance of itself, we extend the bar partially downward so it will geometrically block copies from attaching. 

\begin{figure}[t]
	\centering
	\vspace*{-.2cm}
	\begin{subfigure}[b]{0.45\textwidth}
		\centering	
		\includegraphics[width=.8\textwidth]{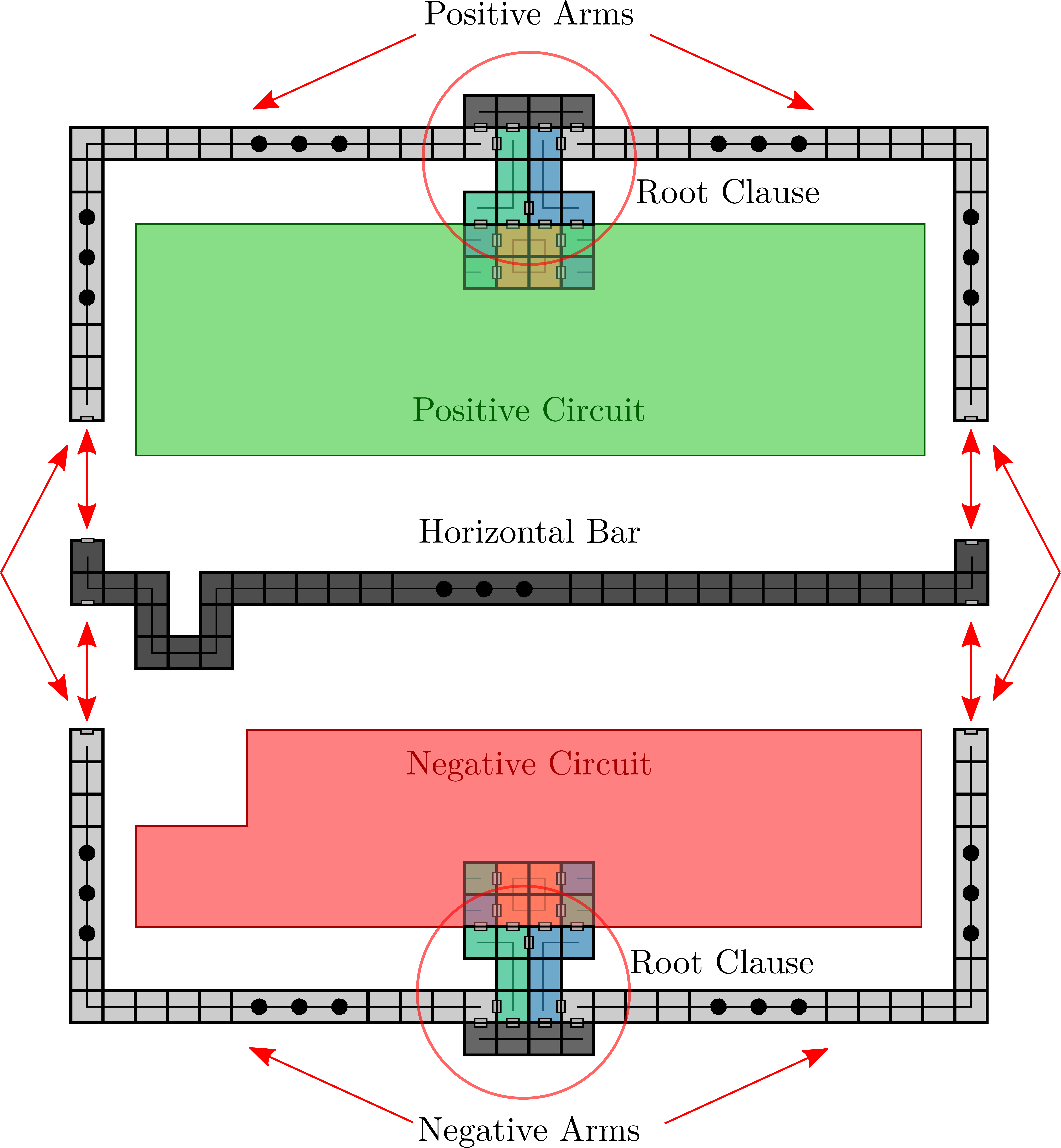}
	    	\caption{Horizontal Bar and Completed Circuits}
		\label{fig:bar}
	\end{subfigure}
	\begin{subfigure}[b]{0.45\textwidth}
		\centering
	    	\includegraphics[width=.7\textwidth]{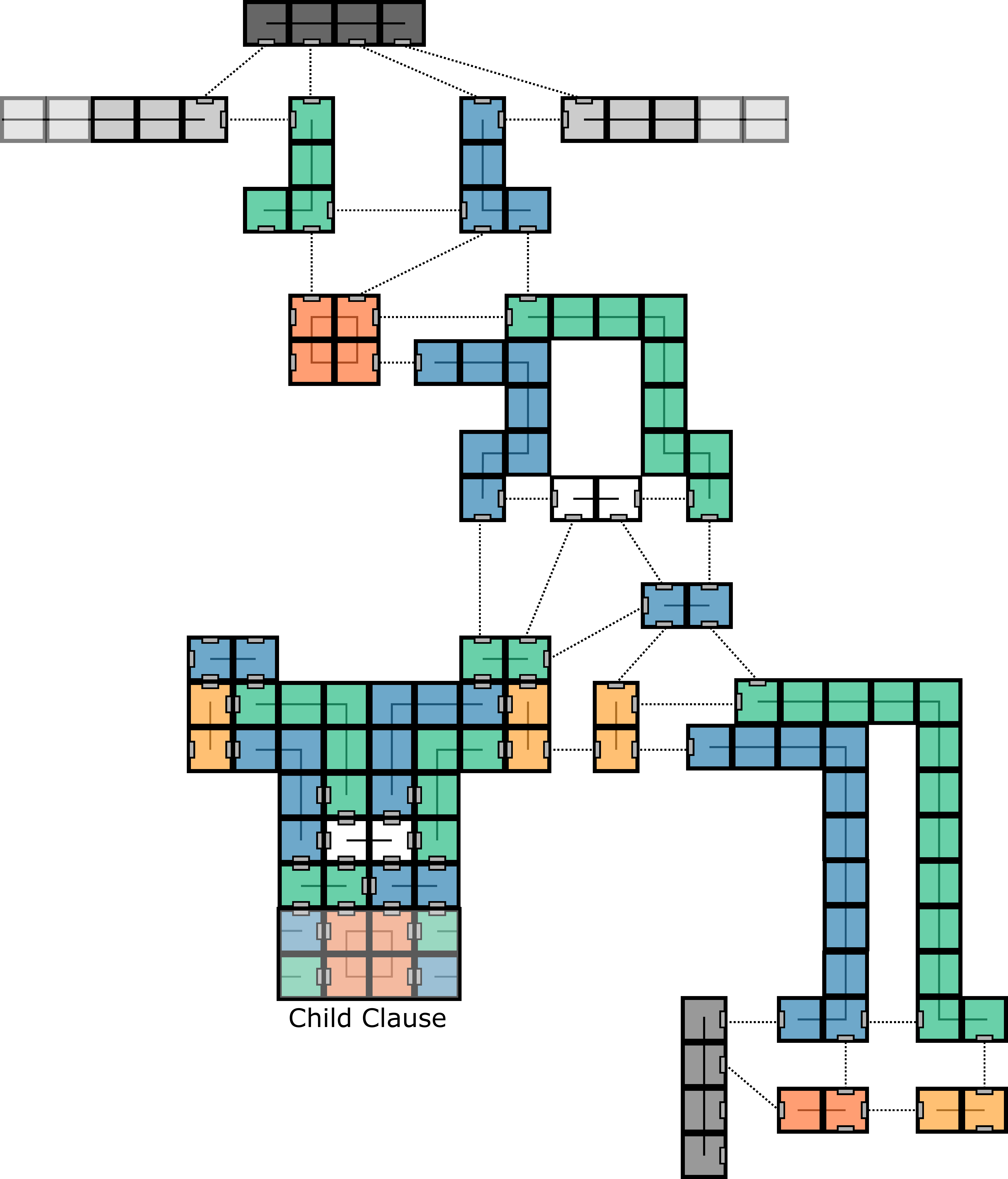}
	    	\caption{Glues on root clause and variable gadgets.}
		\label{fig:rootTerm}
	\end{subfigure}
	\caption{(a) The root clauses and arms joining the positive and negative assemblies with the horizontal bar. The root clause allows a short wire to attach where the arms can then attach. Each arm has a strength-1 glue at the end.  The horizontal bar separates the positive circuit from the negative circuit. The small bump is so the bar can not attach to other horizontal bars.
	(b) Each subassembly of the root clause cannot attach to each other without being satisfied from the child clauses since each subassembly only shares a strength-1 glue with adjacent assemblies. }
\end{figure}

\subsection{Rogue Assemblies} \label{subsec:rogue}
For the construction of the target assembly, each piece is built from the variables up to the root clause. However, the nondeterministic build order means that not all parts of each circuit need to be built in order for the root clause to be satisfied. For instance, if one of the variables in a clause attaches, the OR gates will still allow the wires to attach. Thus, using a variable constitutes setting it to true (and in the negative circuit using a variable is setting the negation to true).

With root clauses satisfied and the arms attaching, a rogue assembly may occur as shown in Figure \ref{fig:rogue}. The corresponding circuit is shown in Figure \ref{fig:roguecircuit}. This can occur because the arms can attach to each other without the horizontal bar. Normally, the variable gadgets would overlap and prevent this attachment if both the positive and the negative circuit used the same variable (which is setting a variable to both true and false). Thus, the positive and negative side each have their own set of variables that make all clauses on their respective sides true. This rogue assembly can only happen if there is a subset for each side that allows all clauses to be true, and thus satisfies the original MP-3SAT-NVP formula.

For an MP-3SAT-NVP instance $\phi$ and an assignment $X_s$ to the variables in $X$, let $A_p$ and $A_n$ be the positive and negative circuit assemblies, respectively, created from $\phi$. We say an assembly $A'_p \sqsubseteq A_p$ represents the assignment $X_s$ if it has attached variable gadgets for the variables in $X_s$ that equal $1$, and has built its root clause. For negated circuits, it must have variable gadgets attached for variables set to $0$ in $X_s$.

\begin{lemma}
For a rectilinear encoding of Monotone Planar 3SAT $\phi$ with neighboring variable pairs and 2HAM system $\Gamma_\phi$ as described above, there exist two producible assemblies $A'_p \sqsubseteq A_p$  and $A'_n \sqsubseteq A_n$ that both represent the same assignment $X_s$ to the variables $X$, if and only if $X_s$ satisfies $\phi$.
\label{lem:rogue}
\end{lemma}

\begin{proof}
If there exists a satisfying assignment $X_s$ to $X$, we may build $A'_p$ by taking the variable gadgets for variables assigned to $1$ and grow the circuit off of them. Since we know all the clauses are satisfied, each clause gadget (including the root clause) may grow resulting in an assembly $A'_p$ that represents $X_s$. By the same argument we know $A'_n$ is producible since $X_s$ satisfies $\phi$, which includes the negated clauses. 

We prove these assemblies are producible only if $X_s$ satisfies $\phi$ via contradiction.
Assume $X_s$ does not satisfy $\phi$, but both assemblies $A'_p$ and $A'_n$ are producible.
Since $X_s$ does not satisfy $\phi$, there must exist at least one unsatisfied clause $c_i$. W.L.O.G., assume $c_i$ is a positive clause. We show the assembly $A'_p$ cannot be produced.

Starting with the case that $c_i$ is the root clause, assume all of the children of $c_i$ are satisfied.
The center input of the clause is a producible subassembly of $A'_p$ since variable gadgets are allowed that satisfy the clauses below it.
We can see in Figure \ref{fig:rootTerm} the other producible subassemblies of the gadgets only have a strength-$1$ glue between them.
This means none of the subassemblies are able to attach to each other on their own.
In order for the arms to attach to the output wire of the root clause, at least one of the AND gates must be fully constructed.
The AND gate cannot assemble unless both halves of the gate have been constructed.
The middle input is built, but the other half of the AND gadget must grow off a completed wire from the variable gadget.
However, since $c_i$ is not satisfied the variables gadgets which satisfy the formula have not attached so the assembly $A'_p$ cannot build the clause gadget.

If $c_i$ is another parent clause that is not the root. Let the clause $c_j$ be the parent clause of $c_i$. If the clause gadget for $c_i$ is not constructed then the gadget for $c_j$ is not buildable.
Since the gadgets used are the same as the root clause, the output wires of the clause gadget for $c_i$ cannot be built without a variable gadget which satisfies the formula.
The middle input of clause gadget representing $c_j$ will not be buildable since this would be the output wire of $c_i$. The middle input goes to two AND gates that cannot construct unless both wires have been built. Thus, the output wire of $c_j$ cannot be built without its children clauses satisfying it. In the case $c_j$ has multiple children, the output wires of all its children are joined by AND gates that will not construct without both inputs.

Finally, consider the case where $c_i$ is a clause without children. In order for the clause's output wire to complete, it must be attached to an OR gadget and the outer wire of the variable gadget. The OR gadget may only attach to a completed wire from a variable gadget (or another OR). The variable gadget cannot be completed without placing the bump, so we cannot have built the outwire of $c_i$. By the same argument as the previous case, this clause not being built results in its parent not being built.

If $c_i$ is not satisfied, the clause gadget for $c_i$ cannot be constructed, which means the assembly $A'_p$ is not producible. 
\end{proof}

\begin{figure}[t]
	\centering
	\begin{subfigure}[b]{0.34\textwidth}
		\centering
	    	\includegraphics[width=1.\textwidth]{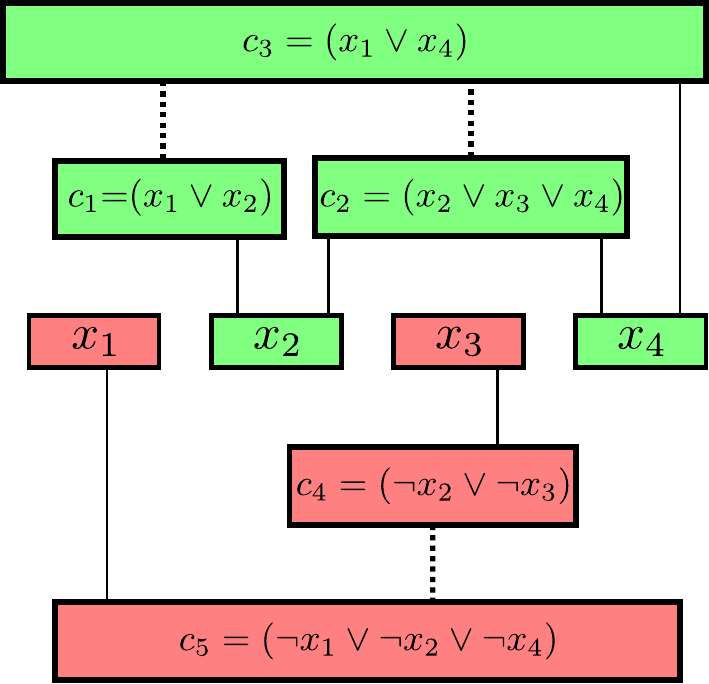}
	    	\caption{Satisfying Assignment}
		\label{fig:rlSAT}
	\end{subfigure}
	\begin{subfigure}[b]{0.3\textwidth}
		\centering
		\includegraphics[width=.7\textwidth]{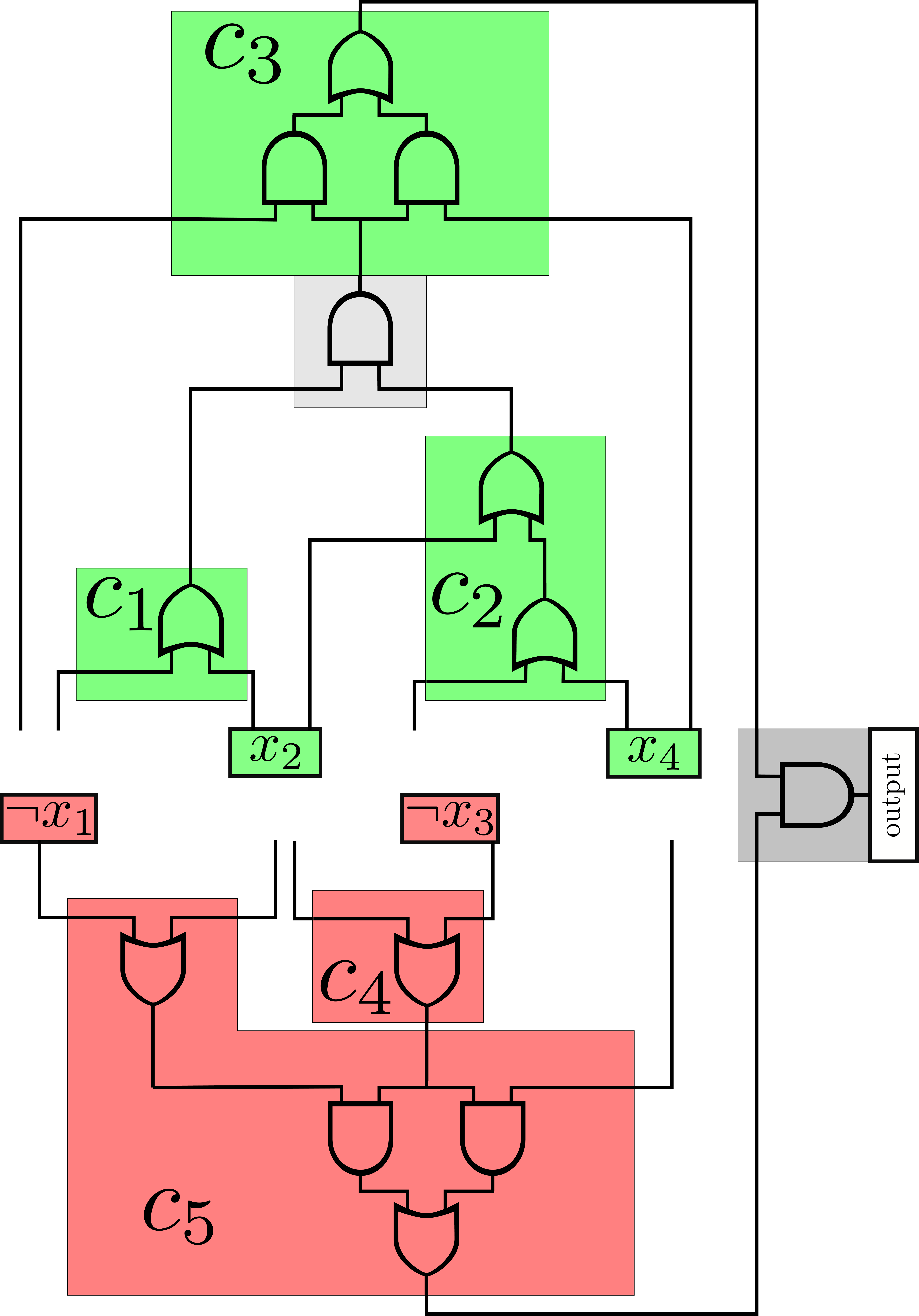}
		\caption{Rogue Assembly Circuit}
		\label{fig:roguecircuit}
	\end{subfigure}
	\begin{subfigure}[b]{0.32\textwidth}
		\centering
		\includegraphics[width=.77\textwidth]{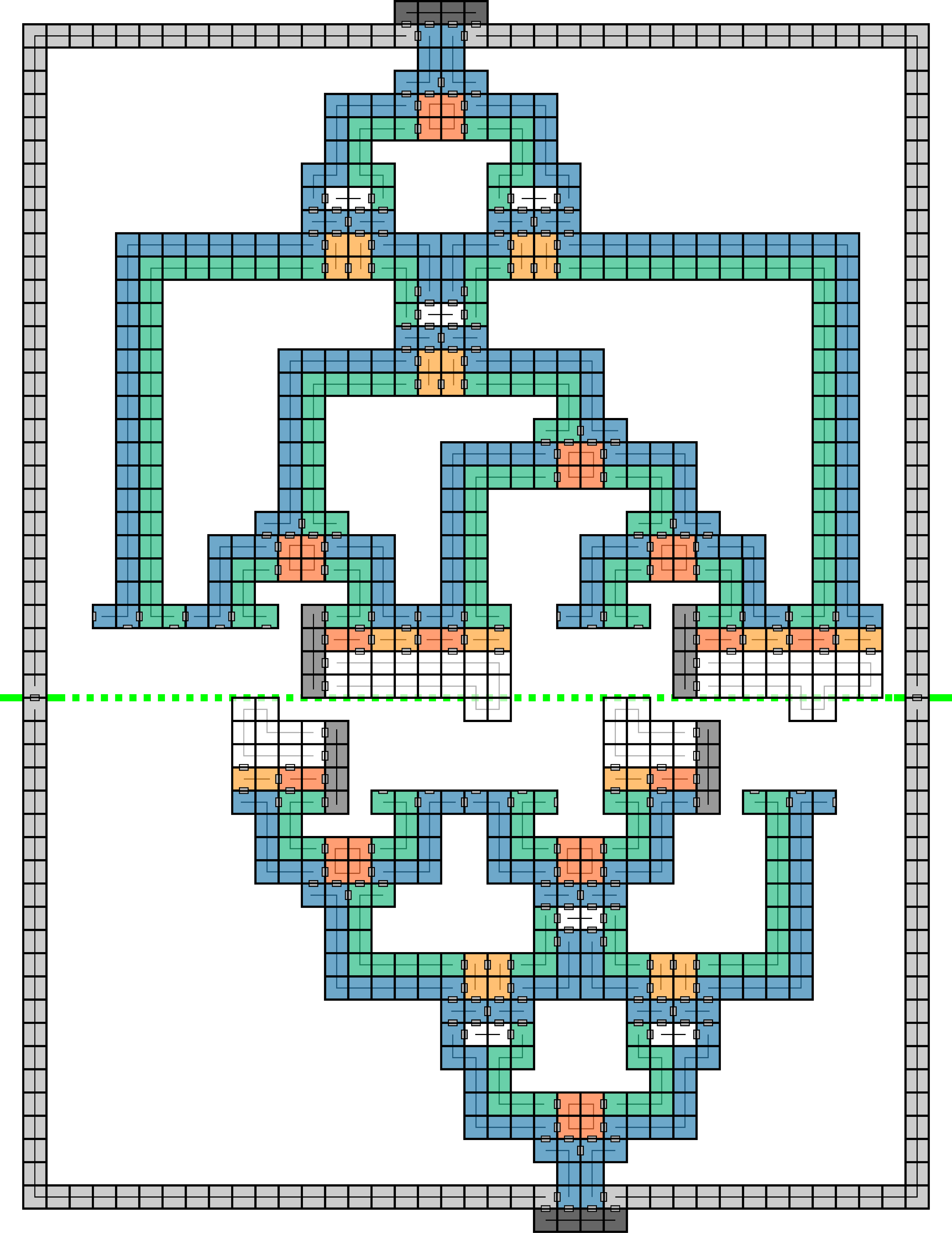}
		\caption{Rogue Assembly}
		\label{fig:rogue}
	\end{subfigure}
	\caption{(a) There exists a satisfying assignment for the example instance with green blocks representing variables which equal $1$ and red blocks representing $0$. 
	(b) Rogue assembly drawn as a circuit with selected variables.
	 (c)  The 2HAM system will produce a Rogue Assembly from the two circuit assemblies which represent the satisfying assignment. 
	 }
\end{figure}

\begin{theorem}
The Unique Assembly Verification problem in the 2HAM is coNP-Complete with $\tau = 2$.
\label{thm:coNPC}
\end{theorem}
\begin{proof}
Given an instance of a rectilinear encoding of Monotone Planar 3SAT with neighboring variable pairs $\phi$, we create a 2HAM system $\Gamma = (T, 2)$ and an assembly $A$ such that $\Gamma$ uniquely produces $A$ if and only if there does not exist a satisfying assignment to $\phi$. We create the assembly $A$ by taking the rectilinear encoding of $\phi$, arranging the rectangles on a grid graph, and replacing the rectangles with the given variable and clause gadgets. We also add the arms and horizontal bar.

Assume there exists a satisfying assignment $X_s$ to the variables $X$, for $\phi$. We know by Lemma \ref{lem:rogue}, there exist two producible assemblies $A'_p$ and $A'_n$ that both contain the arms and have complementary bump positions\footnote{Having complimentary bump positions is equivalent to both representing the same assignment}. These two assemblies can cooperatively bind to one another using the two glues on their arms, and thus produce a rogue assembly as in Figure \ref{fig:rogue}. This means a satisfying assignment to $\phi$ implies $\Gamma$ does not uniquely construct $A$.

Now assume $\Gamma$ does not uniquely produce $A$, so there exists some rogue assembly $B$. The only repeated glues in the tile set of $\Gamma$ are the exposed glues on the arms. Any rogue assembly must use these two glues to assemble, and they must be assembled from two subassemblies of the target by Lemma \ref{lem:subRogue}. Let $B$ be producible by combining two assemblies $b$ and $b'$. Since both $b$ and $b'$ are producible assemblies with both their arms, and they can attach to each other, they are not geometrically blocked. This implies they must represent the same assignment and by Lemma \ref{lem:rogue}, this can only be true if the assignment satisfies $\phi$. 
By viewing which variable gadgets are included in the two assemblies, we can identify the satisfying assignment to $\phi$.
Thus, $\Gamma$ will uniquely produce $A$ if and only if there does not exist a satisfying assignment to $\phi$.
\end{proof}

\section{Verification of Tree-Bonded Assemblies}\label{sec:tree}
In this section, we investigate the problem of Unique Assembly Verification with the promise that the target assembly $A$ is \emph{tree-bonded}, meaning the bond graph of the target assembly forms a tree. Figures \ref{fig:treeShape} and \ref{fig:treeBond} show examples of tree-bonded assemblies. Figure \ref{fig:notTree} shows an assembly whose bond graph contains a cycle and thus is not a tree-bonded assembly. We first present a $\mathcal{O}(|A|^5)$ algorithm for temperature $2$ systems, and then extend this method to provide a $\mathcal{O}(|A|^5 \log \tau)$ time dynamic programming algorithm for the case where the temperature $\tau$ of the system can be passed as a parameter.  Before describing the algorithms, we first introduce some required definitions and the problem formulations.

\begin{figure}[t]
	\centering
	\begin{subfigure}[b]{0.25\textwidth}
		\centering
		\includegraphics[width=.3\textwidth]{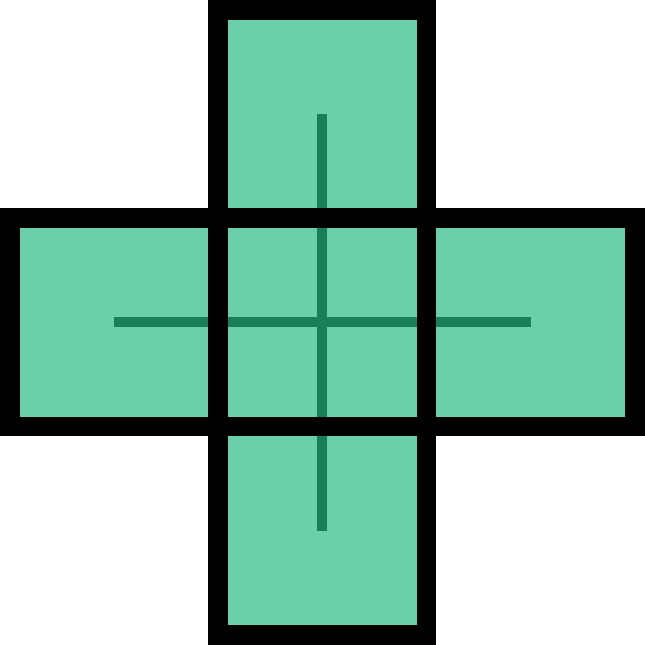}
		\caption{Tree-bonded and Shaped}
		\label{fig:treeShape}
	\end{subfigure}
	\begin{subfigure}[b]{0.25\textwidth}
		\centering
		\includegraphics[width=.3\textwidth]{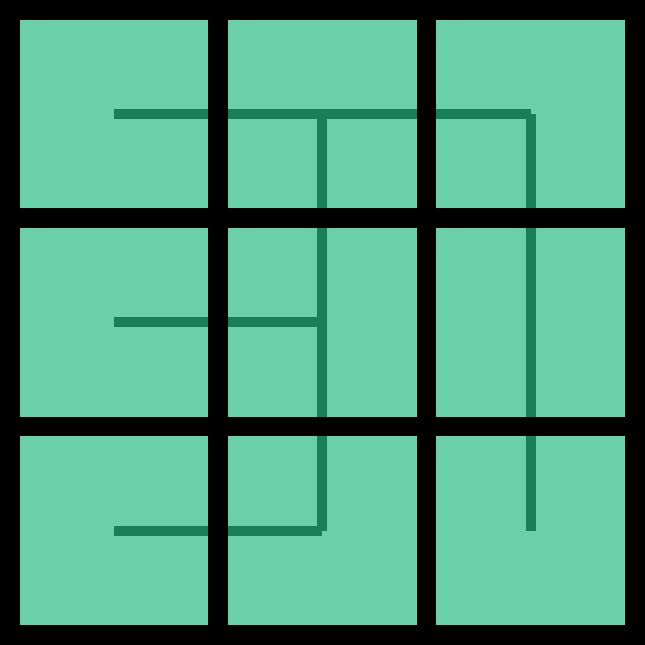}
		\caption{Tree-bonded}
		\label{fig:treeBond}
	\end{subfigure}
	\begin{subfigure}[b]{0.25\textwidth}
		\centering
		\includegraphics[width=.3\textwidth]{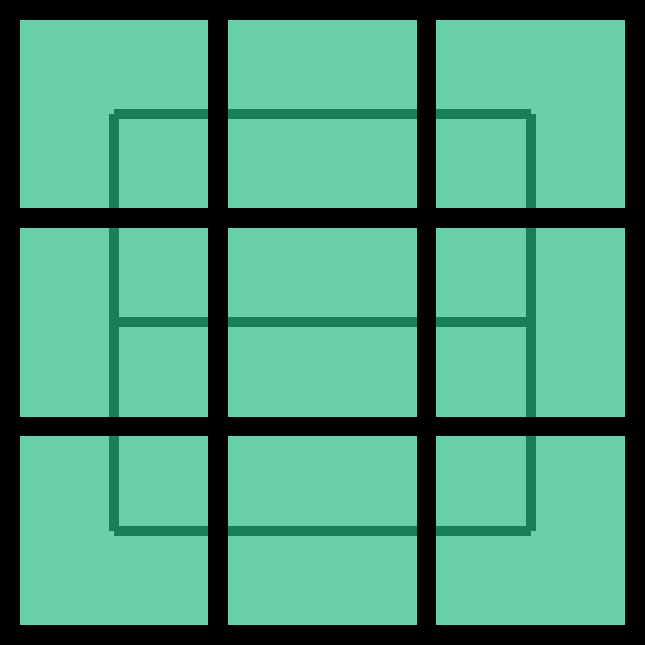}
		\caption{Not Tree-bonded}
		\label{fig:notTree}
	\end{subfigure}
	\caption{(a) Both the shape and bond graph of this assembly are trees. (b) Even though the shape of this assembly is a square, its bond graph still does not contain cycles and thus this assembly is tree-bonded. (c) This assembly is not a tree-bonded graph due to there being a cycle in its bond graph. }
\end{figure}

\paragraph{Tree-Bonded Assemblies.}
An assembly $A$ is a \emph{tree-bonded assembly} if and only if the induced bond graph $\mathcal{G}_A$ is acyclic.

\paragraph{Binding Sites.}  
For two configurations $C_1$ and $C_2$, we say a \emph{binding site} $\mathcal{B}$ is a pair of points $(p_a,p_b)$, such that $||p_a-p_b||_2 = 1$, and the tiles $C_1(p_a)$ and $C_2(p_b)$ have nonzero glue strength between each other. The set of binding sites for two configurations is the set of pairs of points that meet this requirement.
We also define an \emph{inner binding site}.
For two configurations, $C_1$ and $C_2$, and a pair of binding sites $a = (a_1,a'_2), b = (b_1,b'_2)$, let $I(a, b)$ be the set of binding sites that occur on the inside of the loop formed by $a, b$ (inner binding sites). An example of the area enclosed by a loop is seen in Figure \ref{fig:outerAE}.

\paragraph{Simple Loops.}
For two configurations, $C_1$ and $C_2$, and a pair of binding sites $a = (a_1,a'_2), b = (b_1,b'_2)$, we say the loop formed by $a,b$ is a \emph{simple loop} if $|I(a, b)| = 0$.

\paragraph{Origin Configuration}
In discussing different configurations and assemblies, it is useful to anchor a configuration to a fixed point.
For an assembly $A$, the \emph{origin configuration} $A_0$ is the translation of $A's$ configuration such that the bottom left vertex of the bounding box of elements in $dom(A_0)$ is at the origin $(0,0)$.

\begin{problem}[\textsc{Temp2-Tree-UAV}]
\textbf{Input:} A $\tau = 2$ 2HAM system $\Gamma$ and a tree-bonded assembly $A$. \textbf{Output:} Does $\Gamma$ unique produce the assembly $A$?
\end{problem}

\begin{problem}[\textsc{Tree-UAV}]
\textbf{Input:} A 2HAM system $\Gamma$ and a tree-bonded assembly $A$. \textbf{Output:} Does $\Gamma$ uniquely produce the assembly $A$?
\end{problem}

\subsection{Overview}
The high-level goal of this algorithm is to find a rogue assembly that acts as a witness that the instance of UAV is false. We note that a given instance, $P=(\Gamma,A)$, of \textsc{Temp2-Tree-UAV}, where $\Gamma=(T,\tau)$, can be broken down into three possible cases. An example tree-bonded assembly is shown in Figure \ref{fig:treeTarget}.

\begin{enumerate} \setlength\itemsep{0.em}
	\item The instance $P$ is false, and $\Gamma$ produces a tree-bonded rogue assembly.
	\item The instance $P$ is false, and the only rogue assemblies producible in $\Gamma$ are non tree-bonded.
	\item The instance $P$ is true.
\end{enumerate}

At a high level, the algorithm first checks if Case 1 is true and then checks if Case 2 is true. If either are true, the algorithm rejects, otherwise it accepts. Case 1 can be checked efficiently by modifying $\Gamma$ to function as a noncooperative system $\Gamma '$ and utilizing the algorithm for temperature-$1$ UAV provided in \cite{prod2HAM}.
To check the second case, Lemma \ref{lem:subRogue} states that if the instance is false, it suffices to check pairs of subassemblies of the target assembly $A$ in order to find a witness rogue assembly.
Thus, we take two copies of the target assembly and attempt to find possible ways they may bond, even if the resulting assembly places two tiles at the same position. We call the pairs of tiles that contribute glue strength \emph{binding sites}.
Tiles that are in the same position are called \emph{intersections}. An example of both may be seen in Figure \ref{fig:overlapped}.

We first analyze the case of temperature-$2$ systems where only need  two binding sites that do not intersect are needed.
We then generalize this algorithm by using dynamic programming to find the set of binding sites to maximize the binding strength between the assemblies without any intersections.

\begin{figure}[t]
	\centering
	\begin{subfigure}[b]{0.23\textwidth}
		\centering
		\includegraphics[width=.75\textwidth]{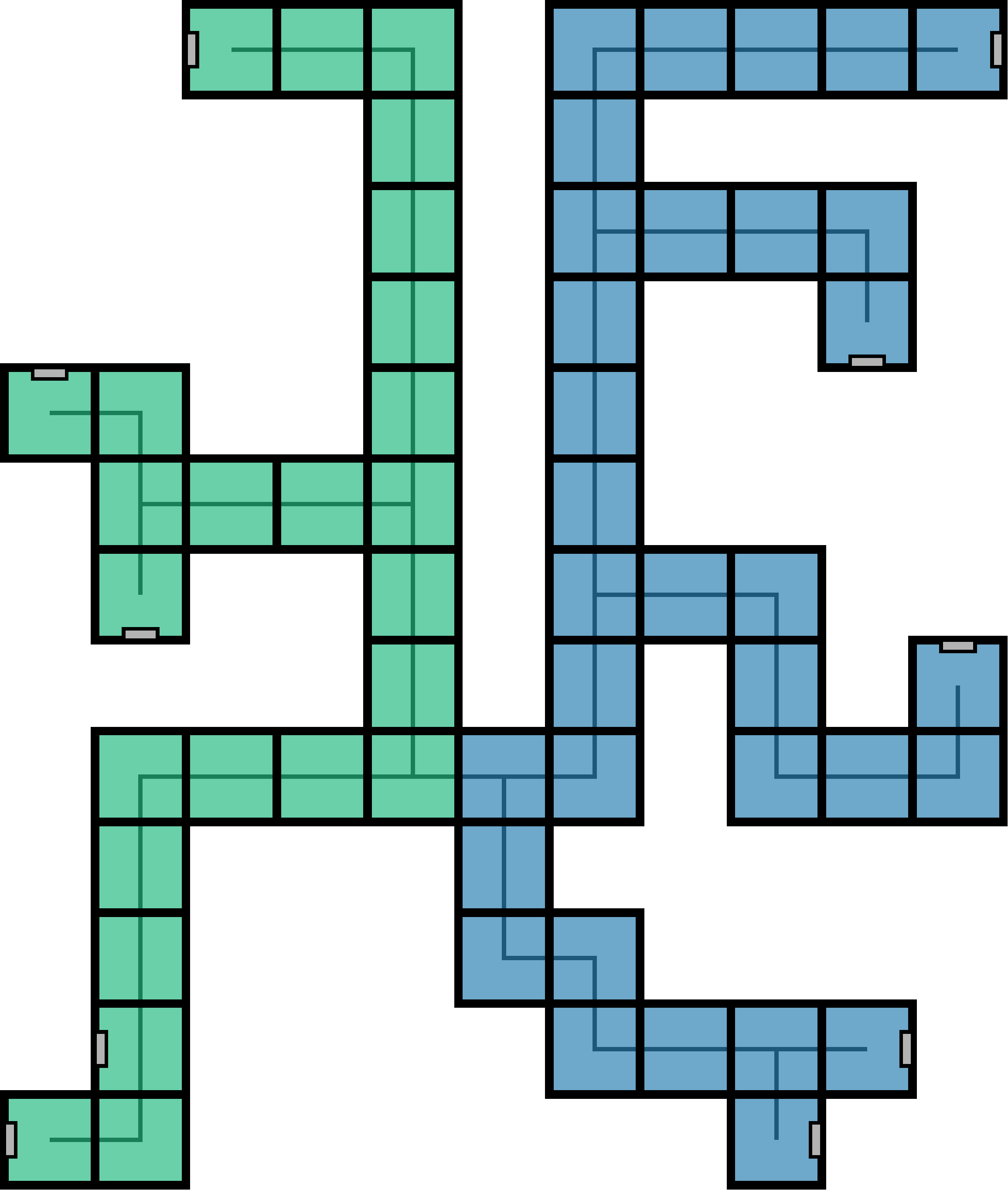}
		\caption{Tree-bonded Assembly}
		\label{fig:treeTarget}
	\end{subfigure}
	\begin{subfigure}[b]{0.24\textwidth}
		\centering
		\includegraphics[width=.85\textwidth]{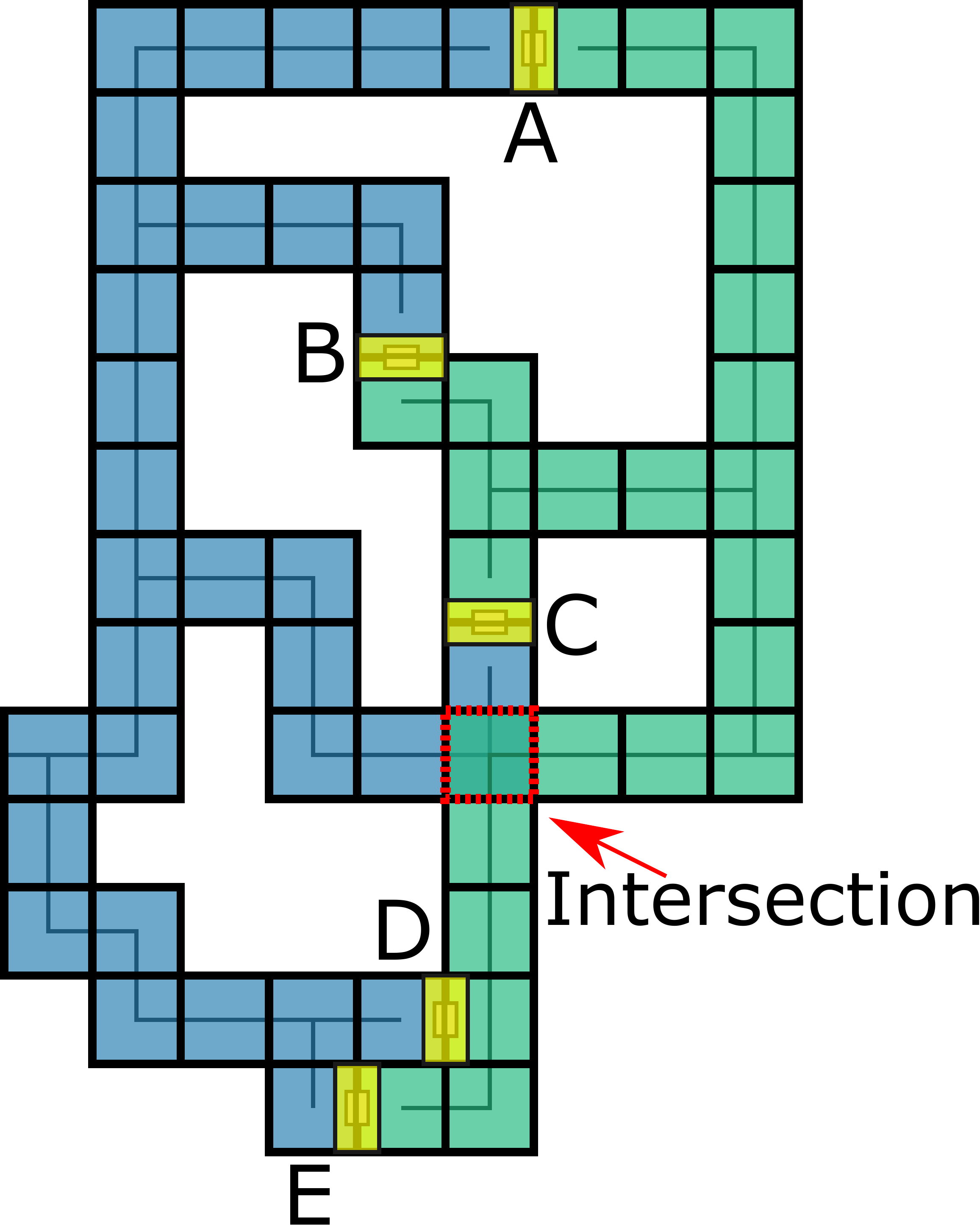}
		\caption{Overlap Subassemblies}
		\label{fig:overlapped}
	\end{subfigure}
	\begin{subfigure}[b]{0.23\textwidth}
		\centering
		\includegraphics[width=.8\textwidth]{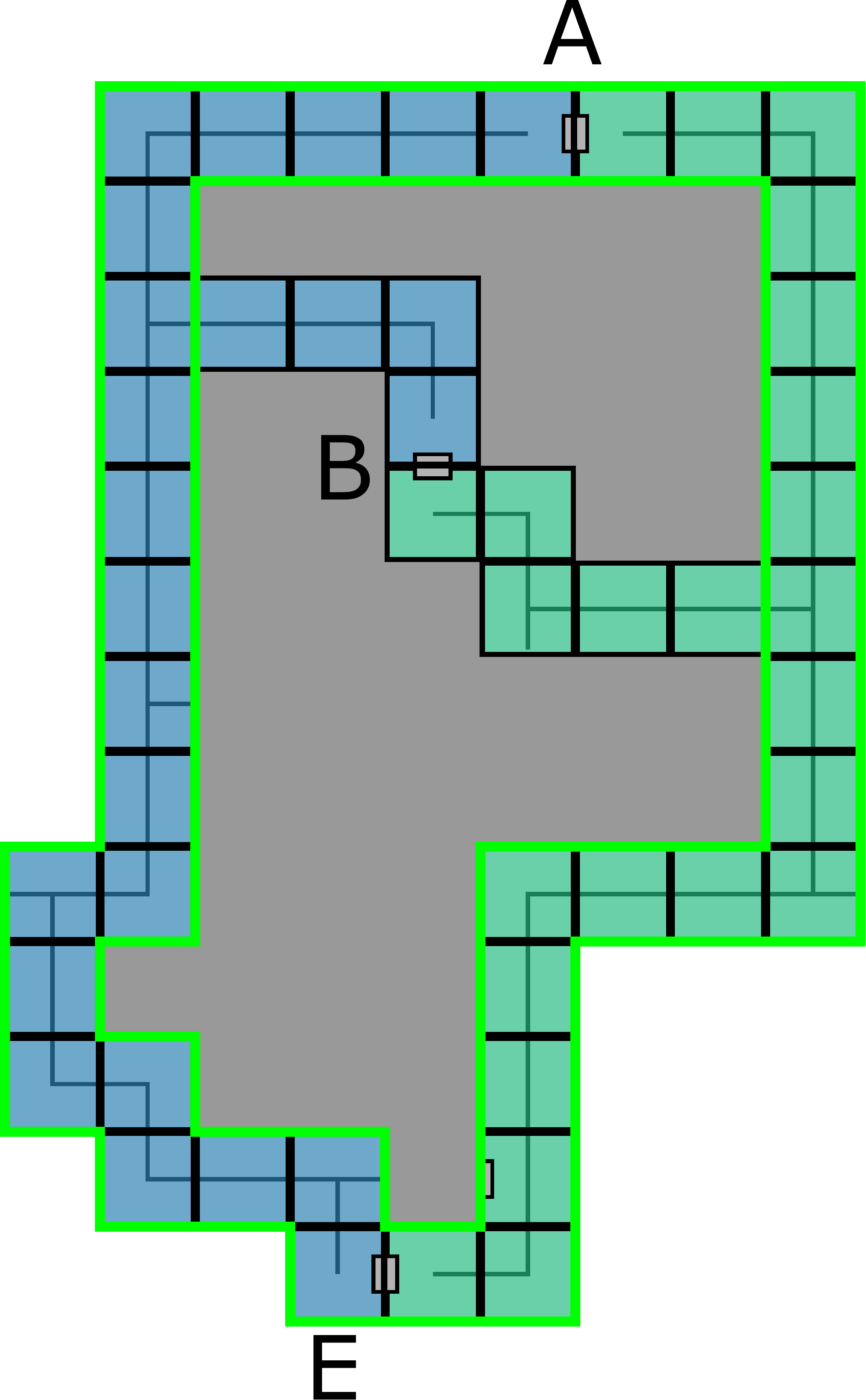}
		\caption{Outer Loop}
		\label{fig:outerAE}
	\end{subfigure}
	\begin{subfigure}[b]{0.23\textwidth}
		\centering
		\includegraphics[width=.8\textwidth]{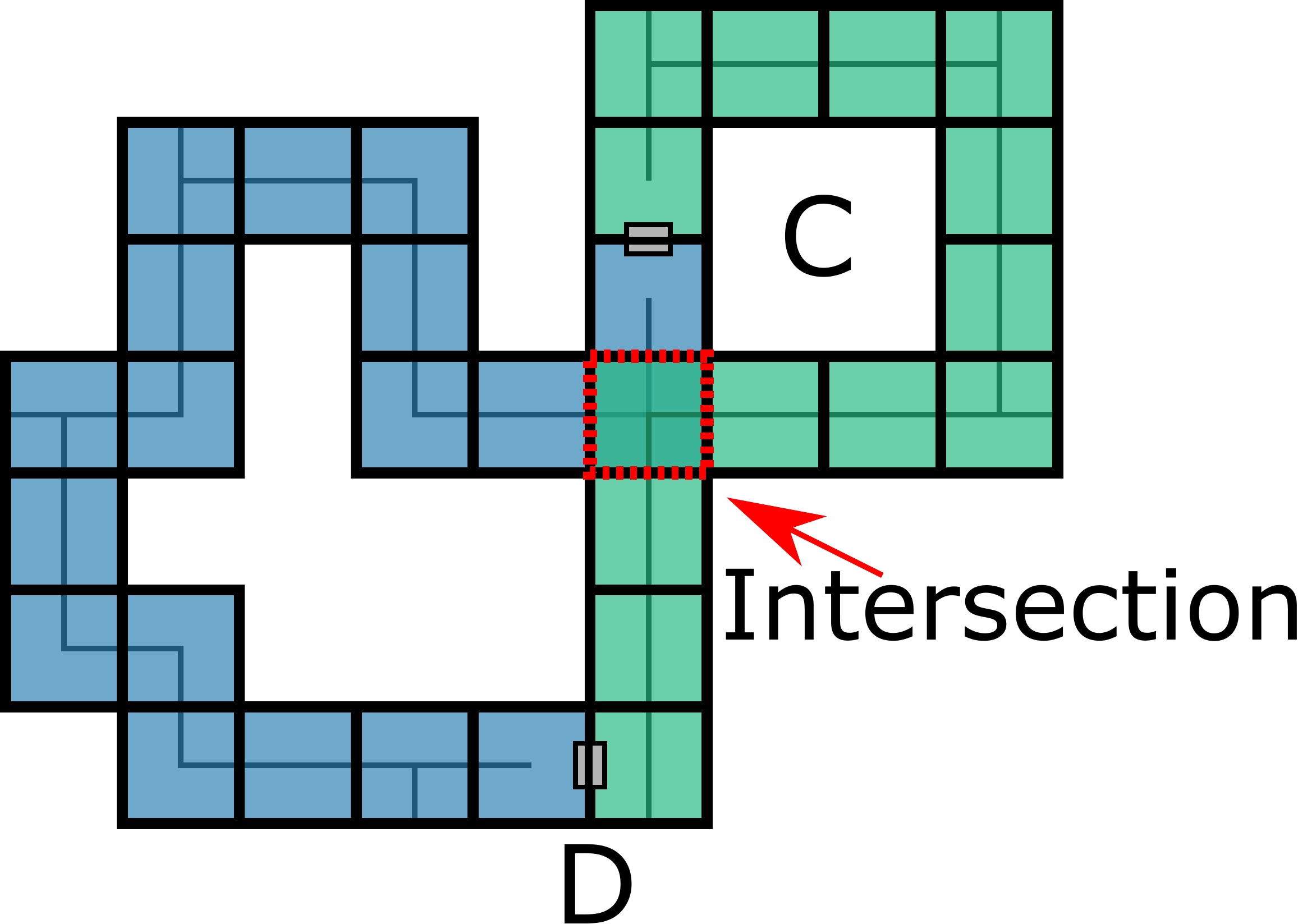}
		\caption{Intersection loop}
		\label{fig:loopCD}
	\end{subfigure}
	\caption{(a) An example tree-bonded target assembly.
	(b) One possible overlap configuration formed by two subassemblies with $5$ binding sites that are highlighted. 
	(c)  The loop formed by binding sites $A$ and $E$ is outlined in green. Any binding site that occurs in the grey shaded area, such as $B$, is in the set of inner binding sites for $(A, E)$. 
	(d) The loop formed using binding sites $(C, D)$ intersects itself and cannot be used. }
\end{figure}

\subsection{Tree-Bonded Rogue Assemblies}
The following algorithm checks if a system uniquely assembles a given shape provided the system is restricted to behaving in a noncooperative manner. This means that two assemblies can only attach if they share one or more strength-$\tau$ glues between them. This system functions equivalently to a temperature-$1$ system where all glues less than strength-$\tau$ are removed and all glues greater than strength-$\tau$ are set to strength-$1$. We modify the system in this way and then use the known polynomial time algorithm for temperature-$1$ UAV \cite{prod2HAM}.

\begin{algorithm}[H]
 \KwData{2HAM System $\Gamma = (T, \tau)$, an assembly $A$}
 \KwResult{Does $\Gamma$ uniquely assembly $A$ if it can only utilize strength $\geq \tau$ glues?}
 	Modify $T$ to create $T '$ by removing all glues of strength less than $\tau$, and setting the strength of all glues of strength $\geq \tau$ to $1$;\\
 	\lIf{\textsc{Temp1-UAV}$(\Gamma' = (T',\tau=1), A)$}{	accept}
 	\lElse{reject}
	
    \label{alg:nonCoopRogue}
 \caption{\textsc{NonCoop-UAV$(\Gamma, A)$}. The runtime of \textsc{Temp1-UAV} is $\mathcal{O}(|A||T|\log|T|)$ \cite{prod2HAM}.}\label{alg:NonCoopUAV}
\end{algorithm}

\begin{lemma}\label{lem:rogueAssembly}
For any 2HAM system $\Gamma = (\Sigma, \tau)$ and tree-bonded assembly $A$, if \textsc{NonCoop-UAV$(\Gamma, A)$} (Algorithm \ref{alg:NonCoopUAV}) is true, and $\Gamma$ does not uniquely assemble $A$, then there exists assemblies $B, B_1, B_2$, s.t. $B \not\sqsubseteq  A$, $B_1,B_2 \sqsubseteq A$, $B_1$ and $B_2$ combine to form $B$ by utilizing cooperative binding.
\end{lemma}
\begin{proof}
Since \textsc{NonCoop-UAV$(\Gamma, A)$} is true, but $\Gamma$ does not uniquely assemble $A$, there must exist a rogue assembly $B' \not\sqsubseteq A$ since any subassembly of $A$ would be tree-bonded.
Consider an assembly tree $\Upsilon_B'$ for $B'$.  Since \textsc{NonCoop-UAV$(\Gamma, A)$} is true, the singleton tile leaves of $\Upsilon_B'$ must be subassemblies of $A$. We will show there exists a node $B \in \Upsilon_B'$, with children $B_1$ and $B_2$, respectively, such that $B' \not\sqsubseteq A$, and $B_1,B_2 \sqsubseteq A$.

Let the root node of the tree $\Upsilon_B$ be the candidate node $B$, and let assemblies $B_1$ and $B_2$ be the two children of $B$. If $B_1$ and $B_2$ are both subassemblies of $A$, then the conditions are met. Otherwise w.l.o.g., assume $B_1 \not\sqsubseteq A$. We now set the candidate node $B$ to $B_1$ and repeat the process. Since all leaves of $\Upsilon_B$ represent subassemblies of $A$, there must be a point in which the candidate $B$ node is some assembly $B \not\sqsubseteq A$, and its children are assemblies $B_1,B_2 \sqsubseteq A$.
\end{proof}

\subsection{Temperature-$2$}
With respect to the given instance of \textsc{Temp2-Tree-UAV} $P$, if $P$ is false, and the algorithm for \textsc{NonCoop-UAV$(\Gamma, A)$} returns `accept', then Lemma \ref{lem:rogueAssembly} implies there exists two subassemblies of the target, $B_1$ and $B_2$, that attach to each other using cooperative binding.

To find these two subassemblies, we start by taking two ``copies'' of the target assembly and finding all $|A|^2$ possible ways to combine the two assemblies, even if it results in intersections.
If any way to combine these assemblies results in at least two binding sites, we  attempt to find $2$-combinable subassemblies. Since we know these subassemblies are also tree-bonded, there only exists one path between each pair of tiles- including the binding sites. So for each pair of binding sites, we take the loop formed by the two binding sites and check if it intersects itself. An example of a loop that intersects itself is shown in Figure \ref{fig:loopCD}. If there ever exists a pair of binding sites whose paths do not intersect, then those two subassemblies will form a rogue assembly and we reject.

\begin{algorithm}[H]
 \KwData{2HAM System $\Gamma = (\Sigma, \tau)$, Tree-Bonded Assembly $A$ with height $h$ and width $w$ }
 \KwResult{Does $\Gamma$ uniquely produce $A$?}

	\lIf{\textsc{NonCoop-UAV$(\Gamma, A)$} rejects}{reject}

	Let $A_0$ be the origin configuration of assembly $A$;\\

	\For{$i\gets -w$ \KwTo $w$ }{

		\For{$j\gets -h$ \KwTo $h$}{

			$A'\gets A_0 + \langle i,j \rangle$;\\
			Let $\mathcal{B}$ be the set of binding sites between $A_0$ and $A'$;\\

			\For{each pair of binding sites $a,b \in \mathcal{B}$}{
				\lIf{The loop formed using $a, b$ does not intersect itself}{reject}
			}

		}

	}
	accept;
  	
    \label{alg:temp2TreeUav}
 \caption{Algorithm to solve \textsc{Temp2-Tree-UAV} in $\mathcal{O}(|A|^5)$ time.}
\end{algorithm}

\begin{theorem} \label{thm:temp2TreeUav}
There is a $\mathcal{O}(|A|^5)$ time algorithm that decides \textsc{Temp2-Tree-UAV}.
\end{theorem}

\begin{proof}
Correctness of Algorithm \ref{alg:temp2TreeUav}: Assume we are given an instance of \textsc{Temp2-Tree-UAV}, $P = ((\Gamma, \tau=2),A)$. The algorithm first checks if $P' = $\textsc{NonCoop-UAV$(\Gamma, A)$} is true. If $P'$ is true then it follows that $A$ is producible in $\Gamma$.

The remainder of the algorithm searches for a rogue assembly by overlapping the target with itself, and checking pairs of binding sites to see if a rogue assembly can be built from the loop formed through them. Since the target is tree-bonded, there only exists a single path between each pair of tiles, including our binding sites. If the two paths formed between the bindings sites do not intersect, or place two different tiles in the same location,  the loop formed by these two paths creates a rogue assembly.

We now argue that if a rogue assembly exists, then the overlapping method will find some rogue assembly. By Lemma \ref{lem:rogueAssembly}, if $P'$ is true, but the instance $P$ is false, there must be a producible rogue assembly $B$ built using cooperative binding, which can be split into two subassemblies $B_1, B_2 \sqsubseteq A$. Now consider two pairs of adjacent tile locations where $B_1$ and $B_2$ share a strength-$1$ glue, and let tiles $t_1,t_2$ and $t'_1,t'_2$ be the tiles at these locations.  Consider the assembly $B'_1$ only composed of the path of tiles from $t_1$ to $t'_1$ in $B_1$, and a second assembly $B_2$ that is the path of tiles from $t_2$ to $t'_2$ in $B_2$. $B'_1$ and $B'_2$ are combinable to build a rogue assembly $B$.

$B'_1$ and $B'_2$ are both subassemblies of $A$, and there is some overlapping of $A$ with an offset of itself in which $B_1$ and $B_2$ are adjacent, and  $t_1,t_2$ and $t'_1,t'_2$ are both in the set of binding sites of this overlap. Since the algorithm checks all pairs of binding sites, it will eventually check that pair, and will attempt to connect the two assemblies that are the simple paths between them, which are $B'_1$ and $B'_2$. Therefore, the algorithm will build the rogue assembly $B' = B'_1 \cup B'_2$ and reject.

Runtime: First \textsc{NonCoop-UAV$(\Gamma, A)$} runs in $\mathcal{O}(|A||T|\log|T|)$ time. While this part of the algorithm becomes the bottleneck if $|T| > |A|$, this cannot be the case. If the size of the tile set is greater than the number of tiles in the assembly, at least one tile can never grow into $A$ and we may reject. Next, for each of the $\mathcal{O}(|A|^2)$ ways to combine $A$ with itself, the algorithm checks at most $\mathcal{O}(|A|^2)$ pairs of binding sites. For each of these pairs, it does a sequence of $\mathcal{O}(|A|)$ time operations to verify the paths do not intersect.
\end{proof}

\subsection{Variable Temperature}
We now present an algorithm for \textsc{Tree-UAV}, a generalization of the previous problem that allows the temperature of the system $\tau$ to be given as input.

This algorithm works in a similar way as the previous algorithm except it does not suffice to only find a single loop since the temperature requirement for assemblies to bind may be greater than $2$. We must find multiple loops between binding sites that do not intersect. Once we have a way to combine the assembly, we will view binding sites and loops in a hierarchical way using inner binding sites. An example of an inner binding site can be seen in Figure \ref{fig:outerAE}. We recursively calculate the max binding strength when taking each pair of binding sites as the outer loop.

After calling \textsc{NonCoop-UAV$(\Gamma, A)$}, we check each possible way to attach $A$ to itself. For each of these ways, we build a $b \times b$ table where $b$ is the total number of induced binding sites. For each pair of binding sites we calculate the maximum value recursively augmented with the table. This allows us to only compute the maximum value once for each loop.

First, we check if the created loop intersects itself.  If it does, we cannot use that loop, so we set the value in the table to be $-1$. Next, we check if the binding sites form a simple loop that does not contain inner binding sites. In this case, the max value is the sum of the glue strength between the binding sites.
For loops that do contain inner binding sites, we perform a  \emph{loop decomposition}, which is the process of breaking a loop into two smaller loops along one of the inner binding sites. An example of a loop being decomposed into simple loops can be seen in Figure \ref{fig:loopDecomp}.
To find the max binding strength of the outer loop, we break the loop up along each inner binding site and recursively get the max strength of the two resulting loops (subtracted by the glue strength of the inner binding site since it would be counted twice). If either of the smaller loops intersects itself, it will return $-1$ and we know not to use that inner binding site.  The max binding strength of the outer loop is then the maximum of these computed values over all choices of inner binding site.  The recursive checks are implemented with a dynamic programming/memorization table to eliminate repeated recursive calls.

\begin{figure}[t]
	\centering
	\begin{subfigure}[b]{0.4\textwidth}
		\centering
		\includegraphics[width=0.9\textwidth]{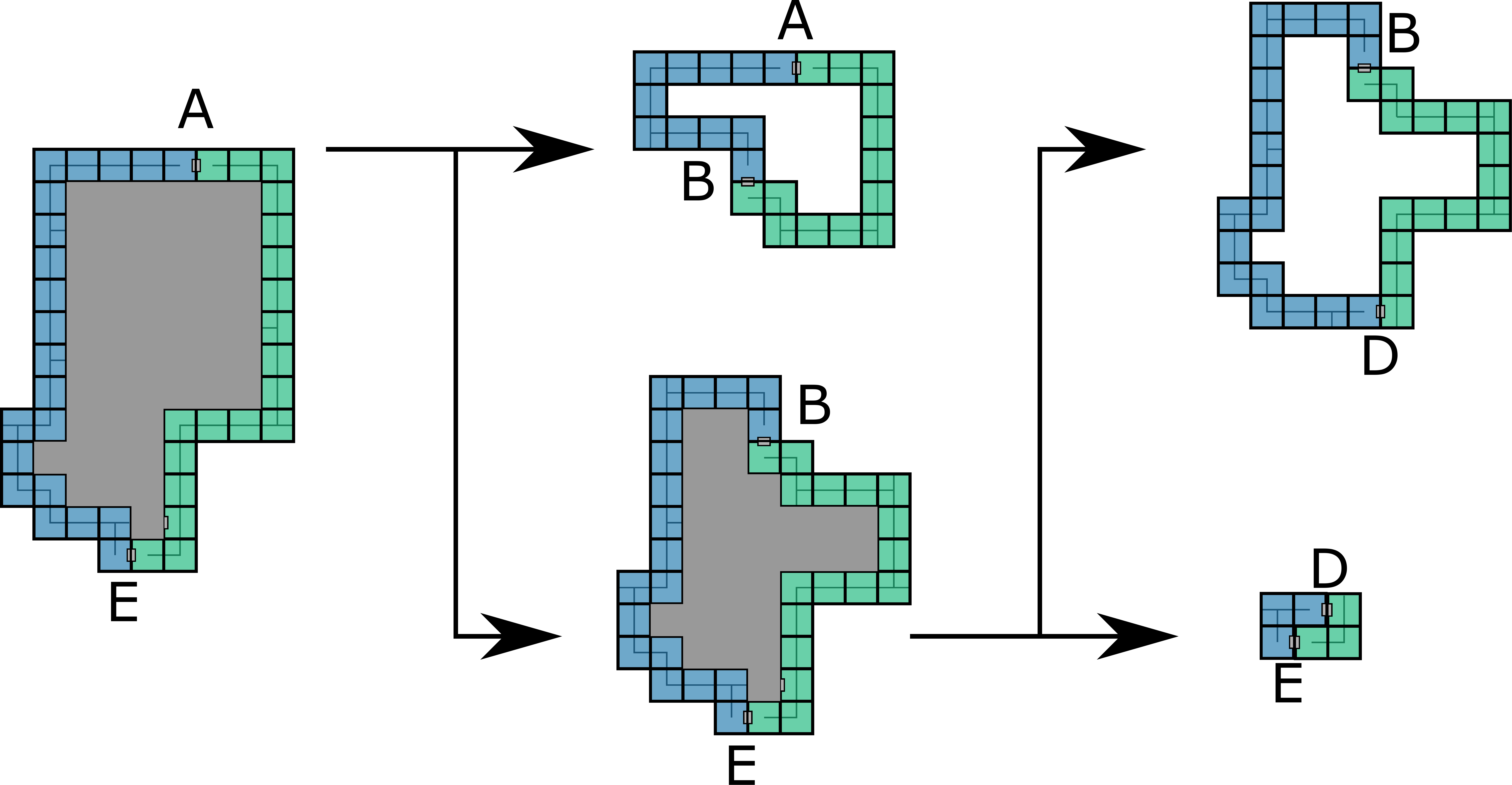}
		\caption{Loop Decomposition}
		\label{fig:loopDecomp}
	\end{subfigure}
	\begin{subfigure}[b]{0.4\textwidth}
		\centering
		\includegraphics[width=.8\textwidth]{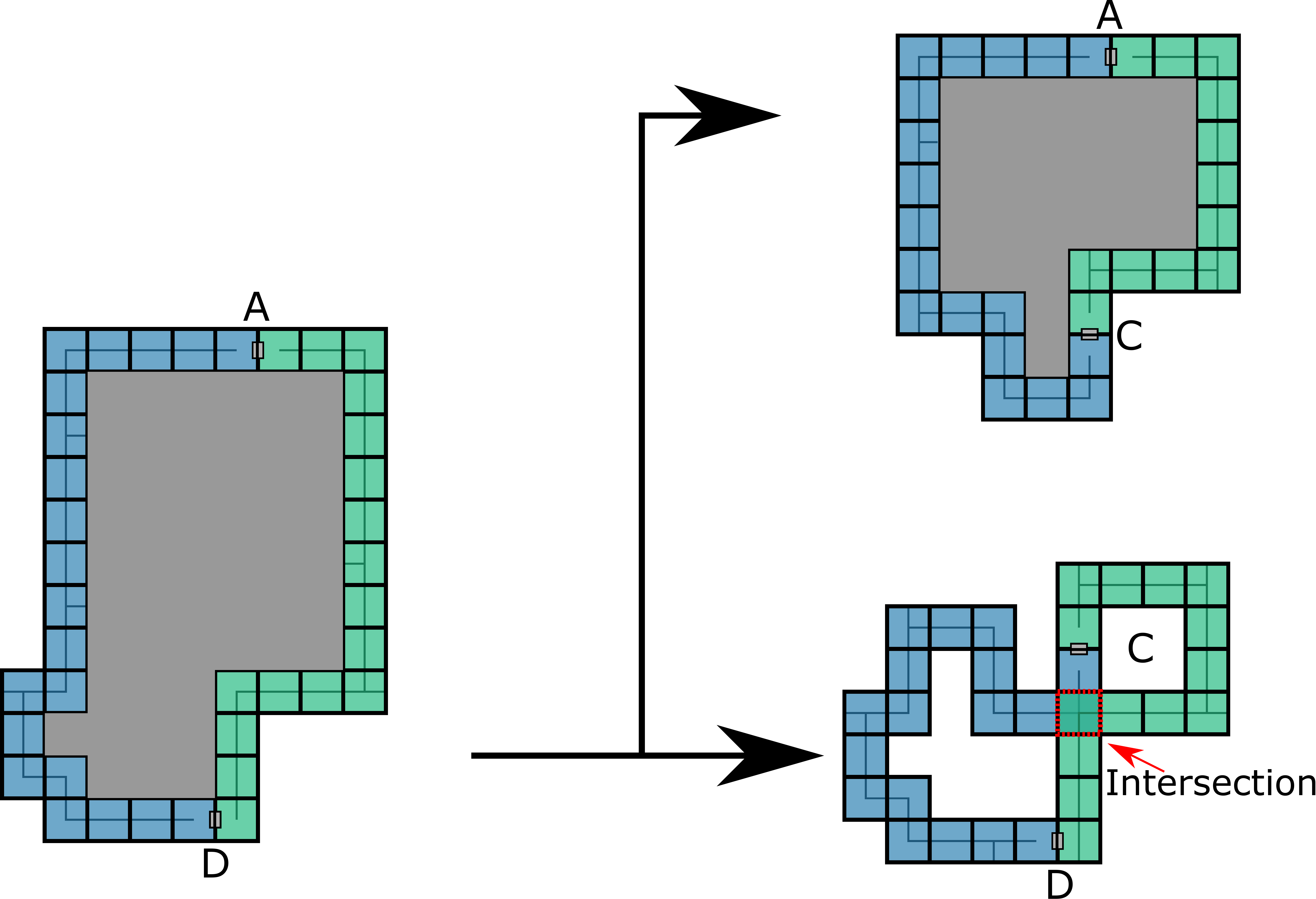}
		\caption{Invalid Loop Decomposition}
		\label{fig:badDecomp}
	\end{subfigure}
	\caption{(a) One possible way to decompose loops into simple loops based on inner binding sites. (b) Decomposing the loop $(A, D)$ along binding site $C$ results in the loop $(C,D)$ which intersects itself. This means we cannot decompose the loop $(A,D)$ along $C$. }
\end{figure}

\begin{algorithm}[H]
 \KwData{2HAM System $\Gamma = (\Sigma, \tau)$, Tree-Bonded Assembly $A$ with height $h$ and width $w$ }
 \KwResult{Does $\Gamma$ uniquely produce $A$?}

	\lIf{\textsc{NonCoop-UAV$(\Gamma, A)$} rejects}{
	reject}

	Let $A_0$ be the origin configuration of assembly $A$;\\

	\For{$x\gets -w$ \KwTo $w$ }{

		\For{$y\gets -h$ \KwTo $h$}{

			$A'\gets A_0 + \langle x,y \rangle$;\\
			Let $\mathcal{B}$ be the set of binding sites between $A_0$ and $A'$;\\
			Let $b = |\mathcal{B}|$;\\
			Create a $b \times b$ table $T_\mathcal{B}$ indexed by the elements of $\mathcal{B}$ with all cells initialized to empty.;\\

			\For{each pair of binding sites $b_1, b_2 \in \mathcal{B}$}{

				\lIf{$maxStr(C, T_\mathcal{B},b_1, b_2) \geq \tau$}{reject}
			}
		}
	
	}
	accept;
  	
    \label{alg:treeUav}
 \caption{Algorithm to solve \textsc{Temp2-Tree-UAV} in $\mathcal{O}(|A|^5)$ time. }
\end{algorithm}

\vspace{1cm}

\begin{algorithm}[H]
 \KwData{Union of two assemblies $C$, Table $T_B$, and Binding sites $b_1, b_2$}
 \KwResult{Does $\Gamma$ uniquely assembly $A$ if it can only utilize strength $\geq \tau$ glues?}
 	\If{$T_\mathcal{B}(b_1,b_2)$ is empty}{
		\lIf{ The loop formed by $b_1, b_2$ intersects itself }{return $-1$}
	
		\lIf{$|I(b_1, b_2)| = 0$}{return $glueStr(b_1) + glueStr(b_2)$}
		Let $T_B(b_1, b_2) = 0$;
	
		\For{$b_i \in I(b_1, b_2)$}{
			\lIf{$maxStr(b_1, b_i)$ or $maxStr(b_i, b_2) = -1$}{continue}
			$s \gets maxStr(b_1, b_i) + maxStr(b_i, b_2) - glueStr(b_i)$;\\
			\lIf{ $s > T_\mathcal{B}(b_1, b_2) $ }{$T_\mathcal{B}(b_1, b_2)  \gets s$}
		}
	}

	 return $T_\mathcal{B}(b_1, b_2)$;

    \label{alg:getStrengthSubroutine}
 \caption{$maxStr(T_\mathcal{B},b_1,b_2)$. The subroutine that calculates the max strength when using two binding sites as the outer loop. The method $glueStr(b)$ takes in a binding site and returns the strength of the glue connecting the two tiles.  }
\end{algorithm}

\begin{theorem}
There is a $\mathcal{O}(|A|^5 \log \tau)$ time algorithm that decides \textsc{Tree-UAV}.
\label{thm:tauTree}
\end{theorem}
\begin{proof}
Correctness: 
We know from Lemma \ref{lem:rogueAssembly}, if \textsc{NonCoop-UAV$(\Gamma, A)$} returns `reject', but the system does not uniquely assemble $A$, there exists a rogue assembly $B$ that may be assembled by two tree-bonded assemblies $B_1$  and $B_2$, which are subassemblies of our target $A$. Using the same method from the previous algorithm, we find the possible ways the target may bind to itself. If any value in the table $T_\mathcal{B}$ is greater than $\tau$, then there exists a set of non-intersecting loops whose union is a producible rogue assembly. We take the union of each of the loops to find the rogue assembly. We know the assembly is producible since none of the paths intersect and the sum of the strength of the binding sites is greater than $\tau$.

Now assume there exists a rogue assembly $B'$ that was not found by \textsc{NonCoop-UAV$(\Gamma, A)$}.  Let $B'$ be a rogue assembly that satisfies Lemma \ref{lem:rogueAssembly}. We know the bond graph of $B'$ must contain a loop. Let the outermost loop of $B'$ be the set of tiles along the outer path of the bond graph, i.e., the loop that contains all other loops.
Since we know $B'$ satisfies Lemma \ref{lem:rogueAssembly}, it was built from two tree-bonded assemblies that do not contain loops. This means the outer loop contains two binding sites we will call $b_1$ and $b_2$.
By filling out $T_\mathcal{B}$,  we calculate the max strength that can be obtained using each pair of binding sites to construct the outer loop. 
Since $B'$ is producible, the sum of the strength of the binding sites must be greater than $\tau$. Thus, the maximum binding strength using $b_1$ and $b_2$ as the outer loop will be greater than $\tau$ as well.

Run Time: We check each of the $\mathcal{O}(|A|^2)$ possible ways to combine the target with itself and create a table if any binding sites are induced. This table will be at most $\mathcal{O}(|A|^2)$-size since we cannot have more binding sites than the size of the assembly, although the number may be substantially smaller. Computing each cell takes $\mathcal{O}(|A|\log \tau)$ time since we must perform integer addition for each pair of binding sites where each integer is less than $\tau$.
\end{proof}

\section{Conclusion}\label{sec:conclusion}

In this paper, we have addressed the long-standing open problem of the complexity of verifying unique assembly within the 2-handed tile self-assembly model and shown that the problem is coNP-complete even at temperature $\tau = 2$ and in two dimensions.  These are the smallest possible values for which this problem can be hard, as both temperature-1 self-assembly and 1-dimensional self-assembly have established polynomial time verification solutions.  Given this hardness, we explored a natural scenario where this problem might be more tractable, and showed that restricting the input assemblies to tree-bonded assemblies allows for an efficient $\mathcal{O}(|A|^5 \log \tau)$-time unique assembly verification algorithm.

\paragraph{Future Work.}
While we have resolved the general question of unique assembly verification in the 2HAM, as well as addressed a natural restricted case of tree-bonded assemblies, there remain important directions for future research.
\begin{itemize}
    \item Our hardness reduction utilizes a tile set that is roughly the size of the input assembly.  In fact, all hardness results in the literature for the 2-handed self-assembly model have this property.  Yet, the computational power of tile self-assembly allows for the self-assembly of large assemblies with much smaller tile sets, as seen in the efficient self-assembly of squares, or the implementation of ``Busy Beaver'' Turing machines~\cite{squares2000}.  This leads to the question of how hard unique assembly verification is when the focus is on large assemblies, but substantially smaller tile sets.  Does the hardness scale with the larger assemblies, or is it tied to the size of the tile sets?  Is there some form of fixed-parameter tractability for the unique assembly verification problem?

    \item We proved that UAV for the multiple tile (or q-tile) model is coNP-complete with polynomial-sized assemblies attaching. Is UAV polynomial in the multiple tile model and the 2HAM in the case where every producible, except the one that grows into the terminal assembly, is bounded by a constant? 

    \item A related question about UAV in the aTAM and the 2HAM is the number of two-handed operations actually required to make the problem hard. If we allowed all subassemblies to grow only by single tile attachments, how many two-handed operations to combine those subassemblies are needed for UAV to remain hard? Does the problem remain hard if only one two-handed operation is allowed? The ability to more efficiently construct shapes by assembling parts separately has been studied in other models as well \cite{effUCI}.
        

    \item Another direction initiated by our efficient tree assembly algorithm is the consideration of other natural restricted classes of the UAV problem.  For example, trees yield efficient verification and have a genus-0 connectivity graph, while our hardness reduction utilizes a high-genus assembly.  How does unique assembly verification scale with respect to the genus of an assembly's connectivity graph?  A related question involves verification for \emph{fully connected} assemblies, a previously-studied concept \cite{stageFully} in which assemblies include positive bonds between all neighboring tiles.  Is UAV still hard under this restriction?

    \item For tree-bonded shapes, does a more efficient and faster UAV algorithm exist?
\end{itemize}

\bibliographystyle{amsplain}
\bibliography{uavBib}

\end{document}